\documentclass[10pt]{amsart}
\usepackage[margin=3.5cm]{geometry}
\usepackage{amssymb}
\usepackage{amsmath}
\usepackage{centernot}
\usepackage{graphicx} 
\usepackage{enumerate}
\usepackage{xspace}
\usepackage{color,xcolor}
\usepackage{enumitem}
\usepackage{amsthm}
\usepackage{dsfont}
\usepackage{mathrsfs}
\usepackage[scr=boondoxo]{mathalfa}
\usepackage{bbm}
\usepackage[colorlinks=true, linkcolor = blue, citecolor = blue, urlcolor = blue]{hyperref}
\usepackage[numbers]{natbib}
\usepackage{microtype}

\theoremstyle{plain}
\newtheorem{theorem}{Theorem}[section]
\newtheorem{corollary}[theorem]{Corollary}
\newtheorem{lemma}[theorem]{Lemma}

\newtheorem*{lemmaOhne}{Lemma}

\theoremstyle{definition}
\newtheorem{discussion}[theorem]{Discussion}
\newtheorem{definition}[theorem]{Definition}
\newtheorem{remark}[theorem]{Remark}
\newtheorem{example}[theorem]{Example}

\newtheorem{SA}[theorem]{Standing Assumption}

\numberwithin{equation}{section}

\DeclareMathOperator{\st}{\s (\l)} 
\newcommand{\rd}{\mathrm{d}}
\newcommand{\vd}{\,\mathrm{d}}

\newcommand{\on}{\operatorname}
\newcommand{\m}{\mathfrak{m}}
\newcommand{\s}{\mathfrak{s}}
\newcommand{\bR}{\mathbb{R}}
\newcommand{\q}{\mathfrak{q}}
\newcommand{\si}{{\on{si}}}
\newcommand{\ac}{{\on{ac}}}
\newcommand{\1}{\mathbbm{1}} 
\newcommand{\cF}{\mathcal{F}}
\renewcommand{\P}{\mathbb{P}} 
\newcommand{\Y}{Y}
\newcommand{\U}{U}
\renewcommand{\S}{S}

\renewcommand{\l}{\alpha}
\renewcommand{\r}{\beta}

\newcommand{\bfFW}{\mathbf F^W}
\newcommand{\bfFY}{\mathbf F^Y}
\newcommand{\bfFZ}{\mathbf F^Z}
\newcommand{\fmu}{\m^U_\ac} 
\newcommand{\fqpp}{\q''_\ac} 

\newcommand\llambda{{\mathchoice
		{\lambda\mkern-4.5mu{\raisebox{.4ex}{\scriptsize$\backslash$}}}
		{\lambda\mkern-4.83mu{\raisebox{.4ex}{\scriptsize$\backslash$}}}
		{\lambda\mkern-4.5mu{\raisebox{.2ex}{\footnotesize$\scriptscriptstyle\backslash$}}}
		{\lambda\mkern-5.0mu{\raisebox{.2ex}{\tiny$\scriptscriptstyle\backslash$}}}}}

\newcommand{\hemail}[1]{\href{mailto:{#1}}{#1}}

\begin{document}
	
	\title[Weak Notions of No-Arbitrage]{On Weak Notions of No-Arbitrage in a 1D General Diffusion Market with Interest Rates}
	
	\author[A. Anagnostakis]{Alexis Anagnostakis}
\address{
    A. Anagnostakis -- Université de Lorraine, CNRS, IECL 
    F-57000, Metz, France; Université Grenoble-Alpes, CNRS, LJK 
    F-38000, Grenoble, France
}
	\email{\hemail{alexis.anagnostakis@univ-grenoble-alpes.fr}}
	
	\author[D. Criens]{David Criens}
	\address{D. Criens -- University of Freiburg, Ernst-Zermelo-Str. 1, 79104 Freiburg, Germany.}
	\email{\hemail{david.criens@stochastik.uni-freiburg.de}}
	
	\author[M. Urusov]{Mikhail Urusov}
	\address{M. Urusov -- University of Duisburg-Essen, Thea-Leymann-Str. 9, 45127 Essen, Germany.}
	\email{\hemail{mikhail.urusov@uni-due.de}}
	
	\keywords{no increasing profit; no strong arbitrage; no unbounded profit with bounded risk; general diffusion; scale function; speed measure; interest rate}
	
	\makeatletter
	\@namedef{subjclassname@2020}{\textup{2020} Mathematics Subject Classification}
	\makeatother
	
	\subjclass[2020]{60J60; 91B70; 91G15; 91G30.}
	
	\thanks{}
	
	\date{\today}
	
	\allowdisplaybreaks
	\frenchspacing
    
	\begin{abstract}
		We establish deterministic necessary and sufficient conditions for the no-arbitrage notions ``no increasing profit'' (NIP),
		``no strong arbitrage'' (NSA) and
		``no unbounded profit with bounded risk'' (NUPBR) in one-dimensional general diffusion markets. These are markets with one risky asset, which is modeled as a regular continuous strong Markov process that is also a semimartingale, and a riskless asset that grows exponentially at a constant rate $r\in \mathbb{R}$.
		All deterministic criteria are provided in terms of the scale function and the speed measure of the risky asset process.
        Our study reveals a variety of surprising effects. For instance, irrespective of the interest rate, NIP is not excluded by reflecting boundaries or an irregular scale function. In the case of non-zero interest rates, it is even possible that NUPBR holds in the presence of reflecting boundaries
and/or skew thresholds. In the zero interest rate regime, we also identify NSA as the minimal no arbitrage notion that excludes reflecting boundaries and that forces the scale function to be continuously differentiable with strictly positive absolutely continuous derivative, meaning that it is of the same form as for a stochastic differential equation.
\end{abstract}
	
	\maketitle

	\section{Introduction}
	\label{sec_intro}

The primary goal of this paper is to investigate several weak notions of no arbitrage for a one-dimensional general diffusion market with a finite time horizon \(T \in (0, \infty)\) and constant interest rate \(r \in \bR\). More specifically, we consider a financial market with a risky asset \(\Y = (\Y_t)_{t \in [0, T]}\) modeled as a one-dimensional semimartingale that is also a regular strong Markov process with continuous paths, a so-called {\em general diffusion semimartingale}, and a bank account \((e^{rt})_{t \in [0, T]}\). 
 The no-arbitrage conditions of interest are NIP (``no increasing proft''), NSA (``no strong arbitrage'') and NUPBR (``no unbounded profit with bounded risk''), which are related as follows:
\[
\text{NUPBR} \implies \text{NSA} \implies \textup{NIP}. 
\]
These conditions are often called ``weak notions'' of no-arbitrage in contrast to their ``strong'' counterparts NA (``no arbitrage'') and NFLVR (``no free lunch with vanishing risk''), cf. the paper of Fontana \cite{Fontana15} for a profound discussion of weak and strong notions of no-arbitrage. The law of a general diffusion is uniquely determined by its {\em scale function} and its {\em speed measure}, collectively known as the {\em diffusion characteristics}, and in this paper we characterize NIP, NSA and NUPBR through the {\em deterministic} diffusion characteristics of \(\Y\).

In the recent paper \cite{CU23}, the second and third named authors investigated NUPBR and the strong notions NA and NFLVR for a one-dimensional general diffusion market {\em without} interest rate. The main results in \cite{CU23} are deterministic characterizations of NUPBR, NA and NFLVR in terms of the diffusion characteristics of \(\Y\). Most notably, it was shown that all three conditions rule out the presence of reflecting boundaries and require the scale function to be continuously differentiable (\(C^1\) for short) with strictly positive
absolutely continuous
derivative, which among other things excludes skewness effects. 
It is interesting to ask whether these structural properties are already enforced by the weaker notions NIP and NSA. 
This question is closely aligned with recent works Buckner, Dowd and Hulley~\cite{BDH23}, Melnikov and Wan~\cite{MelWan21} and Rossello~\cite{rossello12}, which explore market models with reflecting boundaries or skewness effects. 
Specifically, Rossello~\cite{rossello12} examined an exponential skew Brownian motion model and demonstrated that NA fails. Although not stated explicitly, the arguments in \cite{rossello12} show that even NIP fails for this model.
Melnikov and Wan~\cite{MelWan21} analyzed a version of the Bachelier model with a reflecting boundary and observed that NFLVR fails for this market. Lastly, Buckner, Dowd and Hulley \cite{BDH23} studied a reflected geometric Brownian motion model for which they showed that NIP fails.
This clarified false statements in the literature, claiming that
    such a
    model would be in some sense ``free of arbitrage''
    (for details, see \cite{BDH23}).
In their introduction, they also recommend examining weak no arbitrage conditions for other diffusion models from the literature.

In addition to examining weaker notions of no-arbitrage than previously explored in the literature, our setting incorporates an extra dimension by including a non-zero interest rate. In the recent work \cite{anagnostakis2024pricing}, the first named author studied a sticky Brownian motion model and observed that introducing an interest rate has a major influence on no arbitrage properties. Namely, it was shown that the sticky Brownian motion model satisfies NFLVR
if and only if the interest rate is zero. It is natural to ask whether similar effects also occur for the weaker notions NIP, NSA and NUPBR and what consequences interest rates have for the regularity of the scale function and the boundary behavior. We remark at this point that the presence of interest rate is often overlooked in the study of arbitrage.

Having outlined our main motivations, we now describe the most surprising contributions of this paper. 
Let us start with the zero interest rate regime, i.e., the case \(r = 0\).
Given the observations from \cite{BDH23,rossello12} and the results from \cite{CU23} that explain that skewness and/or reflecting boundaries violate the NUBPR, NA and NFLVR conditions (in the zero interest rate regime), it is tempting to conjecture that these features always lead to even strong forms of arbitrage. 
We make the surprising observation that this is not the case in our framework. 

More specifically,
we show that NIP neither excludes reflecting boundaries nor forces the scale function to be absolutely continuous (which is even less than \(C^1\)).
In contrast, still in the case $r=0$,
NSA is equivalent to the absence of reflecting boundaries and a \(C^1\) scale function with strictly positive absolutely continuous
derivative, i.e., it is the minimal notion with these properties. From the viewpoint of these structural conditions, NSA appears to be much closer to NUPBR than to NIP. In fact, irrespective of the interest rate, we will see that NSA and NUPBR are equivalent in case \(\Y\) has no absorbing boundaries (i.e., allowing inaccessible and/or reflecting boundaries), while NIP remains strictly weaker.

In the broader context of mathematical finance,
it is also interesting to highlight that NIP
for $r=0$
might hold while \(\Y\) lacks the
\emph{representation property} (RP),
the latter being
often connected to market completeness (cf., e.g., \cite[Section~VII.2.d]{shir}).
On the contrary,
beginning with NSA,
in the zero interest rate regime,
all stronger notions of no arbitrage
entail the RP of the general diffusion semimartingale \(\Y\).
This follows from the characterization of NSA mentioned above together with the main result of \cite{CU24},
which shows that the RP of a general diffusion semimartingale is equivalent to the absolute continuity of its scale function.

Another surprising finding in this paper is that the inclusion of a non-trivial interest rate completely alters the picture. Indeed, if \(r \not = 0\), even the strongest NUPBR condition can hold despite the presence of reflecting boundaries and scale functions that are less regular than \(C^1\). Among other things, we observe that the effects of stickiness and skewness can cancel each other leading to NUPBR when non-zero interest rates are considered.

	Let us also comment on our proofs. The key ingredients are the stochastic characterizations of NIP, NSA and NUPBR via so-called {\em structure conditions}. 
    These distinguish the weak notions NIP, NSA and NUBPR from their strong counterparts NA and NFLVR, which cannot be captured solely through the semimartingale characteristics of the discounted price process
    (the latter fact is due to \cite[Example~4.7]{KaraKard07}).
    By leveraging arguments based on occupation time formulas, local times and using the Lebesgue decompositions of the speed measure and the second derivative measure of the inverse scale function, we are able to analyze the structure conditions, leading to our main results, the deterministic characterizations of NIP, NSA and NUPBR.
    This strategy is fundamentally different from the approach used in the paper \cite{CU23} for NUPBR,
    which relied on establishing the structural properties of the scale function through deep results on separating times of general diffusions that were established in~\cite{CU22}. 

	Lastly, we briefly comment on related literature. There are numerous studies of no-arbitrage conditions in the literature. 
	The closest to us appear to be the papers of Criens~\cite{Criens2018,criens20}, Criens and Urusov~\cite{CU22, CU23}, Delbaen and Shirakawa~\cite{DelbaenShirakawa2002} and Mijatovi\'c and Urusov \cite{MU12b} that also aim for deterministic characterizations. In this context we again mention the papers \cite{anagnostakis2024pricing,BDH23,rossello12}.
	A broad study of NIP and NSA in general diffusion markets as well as NUPBR in the presence of non-zero interest rates appears to be missing. 

	The paper is organized as follows. 
	In Section~\ref{sec_preliminaries} we set up the mathematical framework, define the notions NIP, NSA and NUPBR and recall their structure conditions. 
	Our main results for NIP (resp. NSA, resp. NUPBR) are presented in Section~\ref{sec_NIP} (resp. Section~\ref{sec_NSA}, resp. Section~\ref{sec_NUPBR}). Finally, in Section~\ref{sec_DE} we discuss a variety of examples and compare our results for NIP, NSA and NUPBR.

	\section{The Financial Market and Weak Notions of No-Arbitrage}
	\label{sec_preliminaries}
	
	In this section, we give a precise introduction to our mathematical setting and we recall the notions of ``increasing profits'', ``strong arbitrage'' and ``unbounded profits with bounded risk'' that are under consideration in this paper.

	\subsection{The Financial Market}
	In this paper, we consider a financial market driven by a regular
	continuous strong Markov process, which is alternatively called a
    \emph{general diffusion}.
	A quite complete overview on the theory of general diffusions can be found in the seminal monograph \cite{ItoMcKean96} by It\^o and McKean.
	Shorter introductions are given in \cite{freedman,Kal21,RevYor,RogWilV2}.
	
	As the concepts of scale and speed are crucial for our results, we recall some facts about them without going too much into detail.
	We take a state space \(J \subset \bR\) that is supposed to be a bounded or unbounded, closed, open or half-open interval. A scale function is a strictly increasing continuous function $\s\colon J\to\mathbb R$
	and a speed measure is a measure $\m$ on $(J,\mathcal B(J))$ that satisfies
	$\m([a,b])\in(0,\infty)$ for all $a<b$ in $J^\circ$, where \(J^\circ\) denotes the interior of~\(J\). We define
	$$
	\l\triangleq\inf J\in[-\infty,\infty)
	\quad\text{and}\quad
	\r\triangleq\sup J\in(-\infty,\infty].
	$$
	The values $\s(\l)$ and $\s(\r)$ are defined by continuity (in particular, they can be infinite).
	We also remark that the speed measure can be infinite near $\l$ and $\r$, and that the values $\m(\{\l\})$ and \(\m (\{\r\})\) can be anything in $[0,\infty]$ provided $\l\in J$ and $\r\in J$, respectively. 
	
	Before we proceed, let us mention that speed measures (and semimartingale local times) are not scaled consistently in the literature. In this paper, we use the scaling from the books of Kallenberg~\cite{Kal21} and Rogers and Williams~\cite{RogWilV2}, which is half the speed measure from the monographs of 
    Freedman~\cite{freedman},
    It\^o and McKean~\cite{ItoMcKean96}
    and Revuz and Yor \cite{RevYor}. To give an example, our speed measure of Brownian motion (on natural scale) is simply the Lebesgue measure, while it is twice the Lebesgue measure in \cite{freedman,ItoMcKean96,RevYor}. Similarly, we use the semimartingale local time scaling of Freedman~\cite{freedman}, Kallenberg~\cite{Kal21}, Revuz and Yor~\cite{RevYor} and Rogers and Williams~\cite{RogWilV2}, which is twice the local time of It\^o and McKean~\cite{ItoMcKean96} and Karatzas and Shreve~\cite{KarShr}. Furthermore, it is worth pointing out that our semimartingale local time is always assumed to be the right-continuous one.
	
	We are in a position to explain our financial framework.
	Throughout this paper, we consider a finite time horizon \(T \in (0, \infty)\).
	Let \(\mathbb{B} = (\Omega, \cF, \mathbf{F} = (\cF_t)_{t \in [0, T]}, \P)\) be a filtered probability space with a right-continuous filtration that supports a regular continuous strong Markov process (in the sense of \cite[Section~V.45]{RogWilV2} except that the underlying setting needs not to be the canonical one) \(\Y = (\Y_t)_{t \in [0, T]}\) with state space~\(J\), scale function~\(\s\), speed measure~\(\m\) and deterministic starting value
    $x_0$. As for the starting value, we always assume that
    $$
    \text{either }x_0\in J^\circ\text{ or }x_0\in J\setminus J^\circ\text{ is a reflecting boundary for }Y.
    $$
    We exclude the case of an absorbing starting value $x_0\in J\setminus J^\circ$, since then the process \(\Y\) is simply constant.
	In the above context, the strong Markov property refers to the filtration~\(\mathbf{F}\).
	
		\begin{SA} \label{SA: semi + boundary}
		\(\Y\) is a semimartingale on the stochastic basis \(\mathbb{B}\).
	\end{SA}
    
	The Standing Assumption~\ref{SA: semi + boundary} is not automatically true in our general diffusion setting. For example, if \(B\) is a Brownian motion starting in zero, then $\sqrt{|B|}$ is \emph{not} a semimartingale (\cite[Exercise~VI.1.14]{RevYor}). The semimartingale property of \(\Y\) is solely a property of the scale function \(\s\) or (more precisely, but equivalently) its inverse.
	The following lemma collects some properties that are proved in \cite[Section~5]{CinJacProSha80}. 

    Recall that for an open interval $I\subset\bR$ and a real-valued function \(\mathfrak{f}\colon I\to\bR\) that is the difference of two convex functions on $I$, one can define the second derivative measure \(\mathfrak{f}'' (\rd x)\) by 
	\[
	\mathfrak{f}'' ((x, y]) \triangleq \mathfrak{f}'_+ (y) - \mathfrak{f}'_+(x), \quad x < y\text{ in }I, 
	\]
	where \(\mathfrak{f}'_+\) denotes the right-hand derivative of~\(\mathfrak{f}\). 
	
	\begin{lemma} \label{lem: isf DC}
		Assume that \(\Y\) is a semimartingale. Then, the inverse scale function \(\q \triangleq \s^{-1}\) is the difference of two convex functions on the interior \(\s (J^\circ)\). Furthermore, in case \(J = [\l, \infty)\) and \(\l\) is absorbing for \(\Y\), it holds
		\[
		\int_{\s (\l) +} (x - \s (\l)) \, |\q''| (\rd x) < \infty. 
		\]
		In case \(J = [\l, \infty)\) and \(\l\) is reflecting for \(\Y\), the second derivative measure \(\q'' (\rd x)\) can be identified with a finite signed measure on every interval \([\l, z]\) with \(z \in (\l,\infty)\). 
	\end{lemma}
    
	\begin{proof}
		These statements follow directly from the discussion in \cite[Section~5]{CinJacProSha80}. 
	\end{proof}
	Of course, suitable adjustments of the last two statements from Lemma~\ref{lem: isf DC} hold also for more general state spaces~\(J\). 

    \smallskip
	In the following, our financial market is supposed to contain one risky asset that is given by the general diffusion semimartingale \(\Y\). 
	Furthermore, we fix a deterministic constant interest rate \(r \in \bR\). The discounting will be done by the usual bank account process \(e^{r t}\) for \(t \in [0, T]\), leading to the discounted price process \(\S = (\S_t)_{t \in [0, T]}\) that is given by 
	\[
	\S_t \triangleq e^{- rt} \, \Y_t, \quad t \in [0, T].
	\]
	We proceed by recalling the notions of no-arbitrage that are under consideration in this paper.

\subsection{Increasing Profits, Strong Arbitrage and Unbounded Profits with Bounded Risk} \label{sec: NIP, NSA, NUPBR intro}
In this section, we recall three rather weak notions of no-arbitrage
\begin{itemize}
\item
\emph{no increasing profit} (NIP),

\item
\emph{no strong arbitrage} (NSA) and

\item
\emph{no unbounded profit with bounded risk} (NUPBR)
\end{itemize}
together with characterization results for them, i.e.,
their corresponding \emph{fundamental theorems of asset pricing} (FTAPs).

The NIP condition is similar to the ``no unbounded increasing profit'' condition that was introduced by Karatzas and Kardaras in \cite{KaraKard07}.
Our presentation follows Fontana \cite{Fontana15}.
In the sequel we use the notation $L(\S)$ for the set of all predictable processes that are integrable w.r.t. the continuous semimartingale \(\S\).
The elements $H\in L(\S)$ are alternatively called \emph{strategies}.

\begin{definition}
A strategy \(H \in L(\S)\) is called an {\em increasing profit} if 
\begin{equation}
\begin{aligned}
\P\text{-a.s.},\;
\text{for all}\; 0\le s \le t \le T:\qquad 
\int_0^s H_u \vd\S_u \leq \int_0^t H_u \vd\S_u
\end{aligned}
\end{equation}
and 
\[
\P\Big( \int_0^T H_u \vd\S_u > 0 \Big) > 0.
\]
If there exist no such strategies, we say that the \emph{NIP} condition holds. 
\end{definition}

We recall the FTAP for NIP, which is essentially due to \cite{KaraKard07}.
The formulation in \cite[Theorem~3.1]{Fontana15} is already adapted to our setting.
Denote the canonical decomposition of the semimartingale~\(\S\) by 
\begin{align} \label{eq: notation DMD}
\S = \S_0 + M + A, 
\end{align}
where \(M\) is a continuous local martingale and \(A\) is a continuous process of finite variation, both starting at zero.

\begin{theorem}[FTAP for NIP]\label{theo: FTAP_NIP}
The NIP condition holds if and only if $\P$-a.s. $\rd A\ll\rd\langle M,M\rangle$.
\end{theorem}

The condition ``$\P$-a.s. $\rd A\ll\rd\langle M,M\rangle$'' is known in the literature as \emph{weak structure condition} (WSC),
see Fontana \cite[Section~3]{Fontana15} or Hulley and Schweizer \cite[Section~3]{HS_2010}.
More specifically, WSC means that there exists a predictable process $\theta=(\theta_t)_{t\in[0,T]}$ such that $\P$-a.s.
\begin{equation}\label{eq:210225a2}
\vd A_t = \theta_t \, \vd \langle M, M \rangle_t.
\end{equation}
Any such process \(\theta\) is called an \emph{instantaneous market price of risk} (IMPR), see \cite{Fontana15} or \cite{HS_2010}.
Thus, using this terminology, the FTAP for NIP can be reformulated as follows:
\begin{equation}\label{eq:310125b1}
\text{NIP}
\;\;\Longleftrightarrow\;\;
\text{WSC}
\;\;\Longleftrightarrow\;\;
\text{there exists an IMPR}.
\end{equation}
In relation with this formulation, it is worth noting that the WSC appeared also in the work of Strasser \cite[Theorem 2.2]{S_2005}, where it was shown that the WSC is equivalent to the condition that every non-negative, predictable wealth process of finite variation is constant.

To introduce more no-arbitrage notions, we define by
\[
\mathcal{A}_a \triangleq \Big\{ H \in L (\S) \colon \int_0^\cdot H_s \vd\S_s \geq - a \Big\}, \quad a \in \bR_+,
\]
the set of so-called ``\(a\)-admissible strategies''.
\begin{definition}
A strategy \(H \in \mathcal{A}_0\) is said to generate a
\emph{strong arbitrage opportunity} if 
\[
\P \Big( \int_0^T H_s \vd S_s > 0 \Big) > 0.
\]
If there exist no such strategies \(H \in \mathcal{A}_0\), we say that the \emph{NSA} condition holds. 
\end{definition}

The NSA condition as defined above was introduced by Strasser \cite{S_2005} under the name ``condition \(\textup{NA}_+\)''.
In fact, it is equivalent to the earlier ``no immediate arbitrage'' notion from Delbaen and Schachermayer \cite{Delbaen1995}
(the equivalence is shown in \cite[Lemma 4.1]{Fontana15}).
The following FTAP for NSA is a part of \cite[Theorem 3.5]{S_2005}
(or see \cite[Theorem 4.1]{Fontana15}).

\begin{theorem}[FTAP for NSA]\label{theo: FTAP_NSA}
The NSA condition holds if and only if there exists an IMPR \(\theta = (\theta_t)_{t \in [0, T]}\) such that
\begin{equation}\label{eq:200125a2}
\inf \Big\{ t \in [0, T) \colon \int_t^{t + h} \theta^2_s \vd \langle M, M\rangle_s = \infty, \, \forall \, h \in (0, T - t] \Big\} = \infty\quad\P\text{-a.s.}
\end{equation}
with the usual convention $\inf\emptyset\triangleq\infty$.
\end{theorem}

\begin{remark}\label{rem:090325a1}
(i)
In connection with this formulation of the FTAP for NSA it is worth noting that, obviously, the integral in~\eqref{eq:200125a2} does not depend on the chosen version of the IMPR.
That is, if there exists an IMPR $\theta$ satisfying~\eqref{eq:200125a2}, then every IMPR satisfies~\eqref{eq:200125a2}.

(ii)
The following point concerning the meaning of~\eqref{eq:200125a2} is worth mentioning.
By the monotone and dominated convergence theorems applied pathwise,
the so-called
\emph{mean-variance tradeoff process}
$K=(K_t)_{t\in[0,T]}$,
$$
K_t=\int_0^t \theta^2_s \vd \langle M, M\rangle_s,
$$
is an increasing process vanishing in the origin
which is continuous and finite on $[0,\zeta)$,
left-continuous at time~$\zeta$
and infinite on $(\zeta,T]$, where
$$
\zeta\triangleq T\wedge\inf\Big\{
t\in[0,T] \colon 
\int_0^t \theta^2_s \vd \langle M, M\rangle_s=\infty\Big\}.
$$
On the event $\{\zeta<T\}$, we have
\begin{itemize}
\item
either $K_\zeta<\infty$
(in which case $K$ jumps to infinity immediately after time~$\zeta$)

\item
or $K_\zeta=\infty$
(which means that $K$ reaches infinity in a continuous way at time~$\zeta$).
\end{itemize}
It is, therefore, tempting to identify~\eqref{eq:200125a2}
with the condition that the mean-variance tradeoff process $K$ does not jump to infinity,
i.e.,
\begin{equation}\label{eq:090325a1}
\P(\zeta<T,K_\zeta<\infty)=0.
\end{equation}
But in fact~\eqref{eq:200125a2} is strictly stronger than~\eqref{eq:090325a1}.
The implication \eqref{eq:200125a2}$\implies$\eqref{eq:090325a1} is straightforward.
The converse is not true because it can happen that
the infimum in~\eqref{eq:200125a2} equals $\zeta<T$ also on paths, for which the process $K$ reaches infinity in a continuous way.
More formally, we will see that, in Example~\ref{ex: x^3} below,
\eqref{eq:090325a1} holds, while~\eqref{eq:200125a2} is violated.
\end{remark}

\begin{definition}
A sequence of trading strategies $(H^n)_{n = 1}^\infty\subset\mathcal A_1$ is said to generate an \emph{unbounded profit with bounded risk} if the sequence of random variables
\[
\Big( \int_0^T H^n_s \vd S_s \Big)_{n = 1}^\infty
\]
is unbounded in \(\P\)-probability.
If there exist no such sequences $(H^n)_{n = 1}^\infty\subset\mathcal A_1$, we say that the \emph{NUPBR} condition holds. 
\end{definition}

The NUPBR condition was introduced under this name and profoundly studied by Karatzas and Kardaras \cite{KaraKard07}.
Earlier, Kabanov \cite{Kabanov_97} had considered the same notion under the name ``BK condition''.
Although not explicitly referenced, NUPBR played a crucial role in the seminal work of Delbaen and Schachermayer \cite{Delbaen1994} on the FTAP for ``no free lunch with vanishing risk''.
Finally, it is worth mentioning that NUPBR is also equivalent to the notions ``no arbitrage of the first kind'', ``no asymptotic arbitrage of the first kind'' and ``no cheap thrill'' that were studied in the literature.
We refer to \cite[Section~5]{Fontana15} and to \cite[Appendix A.1]{KabanovKardarasSong2016} for a detailed account of these notions, precise references and the mentioned equivalences.
The following FTAP for NUPBR is a part of
Choulli and Stricker \cite[Theorem 2.9]{Choulli_Stricker_96}
(or see \cite[Theorem 5.1]{Fontana15} or \cite[Theorem~7]{HS_2010}).

\begin{theorem}[FTAP for NUPBR]\label{theo: FTAP_NUPBR}
The NUPBR condition holds if and only if there exists an IMPR $\theta=(\theta_t)_{t\in[0,T]}$ such that
\begin{equation}\label{eq:310125a2}
\int_0^T \theta^2_s\, \vd \langle M, M\rangle_s < \infty\quad\P\text{-a.s.}
\end{equation}
\end{theorem}

\begin{remark}
If there exists an IMPR $\theta$ satisfying~\eqref{eq:310125a2}, then every IMPR satisfies~\eqref{eq:310125a2} (cf. Remark~\ref{rem:090325a1}~(i)).
\end{remark}

In the literature, ``WSC together with~\eqref{eq:310125a2}'' is known as \emph{structure condition} (SC), see \cite[Section~3]{HS_2010}.
With this terminology at hand, we can recast the FTAP for NUPBR as follows (cf.~\eqref{eq:310125b1}):
$$
\text{NUPBR}
\;\;\Longleftrightarrow\;\;
\text{SC}.
$$
Summarizing all three FTAPs above, we can say that they have a common structure that the no-arbitrage notion is equivalent to the existence of an IMPR plus ``something else''.
For NIP, ``something else'' is empty condition, while
for NSA (resp., NUPBR),
it is~\eqref{eq:200125a2} (resp.,~\eqref{eq:310125a2}).

To conclude our discussion, we observe the implications
$$
\text{NUPBR}
\;\;\implies\;\;
\text{NSA}
\;\;\implies\;\;
\text{NIP},
$$
which are evident both from the definitions and from the FTAPs.
The implications are strict in continuous semimartingale markets
(see \cite{Fontana15} for the examples).
As we will see in Section~\ref{sec: discussion} below, the implications are also strict in our homogeneous diffusion framework (irrespectively of the value of the interest rate \(r \in \bR\)).

	\section{A Deterministic Characterization of NIP}
	\label{sec_NIP}

	In this section we present a deterministic characterization of NIP that only depends on the scale and speed of the general diffusion \(\Y\). Specifically, we will observe intriguing new effects, some of which arise from the presence of a non-zero interest rate.
	By Standing Assumption~\ref{SA: semi + boundary} and Lemma~\ref{lem: isf DC}, on the interior \(\s (J^\circ)\) the inverse scale function \(\q = \s^{-1}\) is the difference of two convex functions. Consequently, the second derivative measure \(\q'' (\rd x)\) is well-defined on \(\s (J^\circ)\). 
	By Lebesgue's decomposition theorem, there exists a unique decomposition
	\[
	\q'' (\rd x)= \fqpp (x) \vd x + \q''_\si (\rd x)\quad
        \text{on }\mathcal B(\s(J^\circ)),
	\]
    where \(\q''_\si\) is a signed measure that is singular w.r.t. the Lebesgue measure \(\llambda\). For $\llambda$-a.a. $x\in\s(J^\circ)$,
the second derivative $\q''(x)$ of $\q$ at the point $x$ exists,
is finite and $\q''(x)=\fqpp(x)$.
Therefore, in what follows, we prefer to write $\q''(x)$ instead of $\fqpp(x)$.

	It is well-known that the process \(\U \triangleq \s (Y)\) is a diffusion on natural scale (i.e., up to increasing affine transformations, the scale function is the identity) and that its speed measure is given by \(\m^{\U} \triangleq \m \circ \s^{-1}\), cf. \cite[Exercise~VII.3.18]{RevYor}. 
	We denote the Lebesgue decomposition (w.r.t. the Lebesgue measure) of the speed measure \(\m^{\U}\) by 
	\[
	\m^{\U}(\rd x) = \fmu (x) \vd x + \m^{\U}_\si(\rd x)\quad
        \text{on } \mathcal{B}(\s(J^\circ)).
	\]
The following is the main result of this section. 
	
	\begin{theorem} \label{theo: main_NIP}
		The NIP condition is satisfied if and only if the following conditions~hold:
		\begin{enumerate}
			\item[\textup{(i)}]
			Every accessible boundary point \(b \in J \setminus J^\circ\) satisfies one of the following two conditions: 
			\begin{enumerate}
				\item[\textup{(i.a)}] \(b\) is absorbing and either \(r = 0\) or \(b = 0\);
				\item[\textup{(i.b)}] \(b\) is reflecting and 
				\[
				r b\, \m^\U (\{ \s (b) \}) = \begin{cases} \frac12 \q'_+ (\s (\l)), & b = \l, \\[1mm] \frac12 \q'_- (\s (\r)), & b = \r. \end{cases}
				\]
			\end{enumerate}
			\item[\textup{(ii)}]
			\(r \q (x) \m^{\U}_\si (\rd x) = \frac{1}{2} \q''_\si (\rd x)\) on \(\mathcal{B}(\s (J^\circ) )\).
			\item[\textup{(iii)}]
			\(r \q (x) \fmu (x) = \frac{1}{2}\q'' (x)\) for $\llambda$-a.a. \(x \in \{z \in \s (J^\circ) \colon \q' (z) = 0\}\).
		\end{enumerate}
        Furthermore, in the case where \textup{(i)}--\textup{(iii)} hold, (a version of) the IMPR is given by the formula
        \begin{equation}\label{eq:200125a1}
        \theta_t=e^{rt}\gamma(\U_t),\quad
        \gamma(x)\triangleq\frac{\frac12\q''(x)-r\q(x) \fmu(x)}{[\q'(x)]^2}\1_{\{\q' \ne 0\}}(x)\quad
        \text{for $\llambda$-a.a. }x\in\s(J).
        \end{equation}
        In particular, for every accessible boundary $b\in J\setminus J^\circ$,
    the value $\gamma(\s(b))$ can be chosen arbitrarily.
        \end{theorem}

In relation with (iii) above it is worth noting that, as $\q$ is the difference of two convex functions, the derivative $\q'$ is well-defined up to a Lebesgue null set (more precisely, even up to a countable set).

    Our main tool in the proof of the following result is the FTAP for NIP that we recalled in Theorem~\ref{theo: FTAP_NIP}. To use it we need to establish the semimartingale decomposition of the price process~\(S\).
    For what follows, recall the notation $S=S_0+M+A$ from~\eqref{eq: notation DMD}.
    For a process \(Z = (Z_t)_{t \in [0, T]}\), we define the hitting time of a point \(x\) by
    \[
    T_x (Z) \triangleq \inf \{t \in [0, T] \colon Z_t = x \}.
    \]
    Furthermore, throughout this paper, for a continuous semimartingale \(Z\), \(\{L^x_t (Z) \colon (x, t) \in \bR \times [0, T]\}\) denotes the right semimartingale local time field of \(Z\). 
    
\begin{lemma} \label{lem: semimartingale decomposition concrete}
Suppose that \(J^\circ = (\l, \infty)\). 
\begin{enumerate}
\item[\textup{(a)}]
In case \(\l\) is inaccessible or absorbing for \(Y\), then 
  	\begin{equation} \label{eq: final DMD}
			\begin{split}
				\rd \langle M, M\rangle_t &= e^{- 2rt} \big[\q'_+ (U_t) \big]^2 \1_{\{t < T_{\s (\l)} (U)\}} \vd \langle U, U\rangle_t, \\ 
				\rd A_t &= e^{-rt}   \Big[ - r \q (U_t) \vd t + \frac{1}{2} \int_{\s (J^\circ)} \rd L^x_t (U) \, \q'' (\rd x) \Big], 
			\end{split}
		\end{equation}
        where the indicator \(\1_{\{t < T_{\s (\l)} (U)\}}\) is included
        to emphasize that we do not require \(\q'_+ (\s(\alpha))\) to be well-defined (and, indeed, the limit of $\q'_+(u)$, as $u\searrow\s(\alpha)$, can fail to exist).

\item[\textup{(b)}] 
In case \(\l\) is reflecting for \(Y\), then 
  	\begin{equation} \label{eq: final DMD - re}
			\begin{split}
				\rd \langle M, M\rangle_t &= e^{- 2rt} \big[\q'_+ (U_t) \big]^2 \vd \langle U, U\rangle_t, \\ 
				\rd A_t &= e^{-rt} \Big[ - r \q (U_t) \vd t+ \frac{1}{2} \q'_+ (\s (\l)) \vd L^{\s (\l)}_t (U) + \frac{1}{2} \int_{\s (J^\circ)} \rd L^x_t (U) \, \q'' (\rd x) \Big].
			\end{split}
		\end{equation}
\end{enumerate}
\end{lemma} 

\begin{proof}
(a)
In the following, we suppose that \(\l\) is inaccessible or absorbing for the diffusion \(Y\). This yields that \(\P\)-a.s.
\[
\rd Y_t = \1_{\{t  <  T_{\l} (Y)\}} \vd Y_t.
\]
Using this identity and integration by parts yields that \(\P\)-a.s.
\begin{align} \label{eq: S dynamics - abs}
\rd S_t = - r S_t \vd t + e^{- rt} \1_{\{t \, < \, T_{\l} (Y)\}} \vd Y_t.
\end{align} 
Next, we identify a formula for \(\1_{\{t  < T_{\l} (Y)\}} \vd Y_t\).
Recall from Lemma~\ref{lem: isf DC} and Standing Assumption~\ref{SA: semi + boundary} that \(\q\) is the difference of two convex functions on \(\s (J^\circ)\). Furthermore, recall from \cite[Theorem~V.49.1]{RogWilV2} that \(\P\)-a.s. \(x \mapsto L^x_t (U)\) is continuous on \(\s (J^\circ)\) for every \(t \in [0, T]\). Hence, by the generalized It\^o formula (\cite[Theorem~IV.45.1]{RogWilV2}) applied to \(Y = \q (\s (Y)) = \q (U)\), we obtain that \(\P\)-a.s.
\begin{align} \label{eq: in lem Y ito}
\rd Y_t = \q'_+ (U_t) \vd U_t + \frac{1}{2} \int_{\s (J^\circ)} \vd L^{x}_t (U) \, \q'' (\rd x), \quad t < T_{\l} (Y) = T_{\s (\l)} (U). 
\end{align} 
Now, combining \eqref{eq: S dynamics - abs} and \eqref{eq: in lem Y ito}, we get that \(\P\)-a.s.
\[
\rd S_t = - r e^{- rt} \q (U_t) \vd t + e^{-rt} \1_{\{t < T_{\s (\l)} (U)\}} \Big( \q'_+ (U_t) \vd U_t + \frac{1}{2} \int_{\s (J^\circ)} \vd L^{x}_t (U) \, \q'' (\rd x) \Big).
\]
As \(U\) is a local martingale (which follows from \cite[Corollary~V.46.15]{RogWilV2} and the assumption that \(\l\) is inaccessible or absorbing), this equation entails \eqref{eq: final DMD} immediately.

\smallskip
(b)
Assume that \(\l\) is reflecting for \(Y\). Then, \cite[Lemma~C.28]{CU22} shows that \(\P\)-a.s.
\begin{align*} 
\rd Y_t = \q'_+ (U_t) \vd U_t + \frac{1}{2} \int_{\s (J^\circ)} \vd L^{x}_t (U)\, \q'' (\rd x), 
\end{align*} 
where its prerequisites are satisfies by Lemma~\ref{lem: isf DC} and Standing Assumption~\ref{SA: semi + boundary}.
Notice that at this point we also use that the semimartingale local time of the diffusion $U$ on natural scale is continuous in the space variable on $\s(J^\circ)$
(but not necessarily at the boundaries),
cf. \cite[Lemma~C.15 and/or Remark~C.16~(a)]{CU22}.
Now, integration by parts yields that \(\P\)-a.s.
\begin{align} \label{eq: lem decomposition; ref S dyn}
\rd S_t = - r e^{- rt} \q (U_t) \vd t + e^{- rt} \Big(  \q'_+ (U_t) \vd U_t + \frac{1}{2} \int_{\s (J^\circ)} \vd L^{x}_t (U) \, \q'' (\rd x) \Big).
\end{align}
Finally, \cite[V.47.23 (ii)]{RogWilV2} shows that the process
		\(
		U - L^{\s(\l)} (U) / 2
		\)
		is a local martingale. In Appendix~\ref{appendix_A: proof_lem: comp} we provide a new proof of this fact, which we believe to be instructive. Using this fact and \eqref{eq: lem decomposition; ref S dyn}, the formulas in \eqref{eq: final DMD - re} follow immediately.
\end{proof} 
	
\begin{proof}[Proof of Theorem~\ref{theo: main_NIP} (Necessity)]
		First, we assume that the NIP condition holds. To ease our presentation, we make the simplifying assumption that \(J^\circ = (\l, \infty)\). The general case is notationally more complex but otherwise similar. In the following we show that the properties (i)--(iii) are satisfied.

\smallskip 
{\em Proof of (i):}
By the FTAP for NIP as restated by Theorem~\ref{theo: main_NIP}, there exists an IMPR \(\theta = (\theta_t)_{t \in [0, T]}\) such that \(\P\)-a.s.
\begin{align} \label{eq: IMPR eq}
\rd A_t = \theta_t \vd \langle M, M\rangle_t.
\end{align} 
Using Lemma~\ref{lem: semimartingale decomposition concrete} and the occupation time formula (\cite[Theorem~IV.45.1]{RogWilV2}),
we obtain that \(\P\)-a.s.
\begin{align} \label{eq: appl IMPREq b - 1}
\1_{\{U_t = \s (\l)\}} \vd \langle M, M \rangle_t = 0.
\end{align}
Again by Lemma~\ref{lem: semimartingale decomposition concrete}, we also get that \(\P\)-a.s.
\begin{align} \label{eq: appl IMPREq b - 2}
e^{rt} \1_{\{U_t =  \s (\l)\}} \vd A_t = - r \l \1_{\{U_t = \s (\l)\}} \vd t + \1_{\{\l \text{ is reflecting for \(Y\)\}}} \, \frac{1}{2} \q'_+ (\s (\l)) \vd L^{\s (\l)}_t (U). 
\end{align} 
Combining \eqref{eq: IMPR eq} with \eqref{eq: appl IMPREq b - 1} and \eqref{eq: appl IMPREq b - 2}, we get that \(\P\)-a.s.
\begin{align} \label{eq: IMPR - pf}
- r \alpha \1_{\{U_t = \s (\l)\}} \vd t + \1_{\{\l \text{ is reflecting for \(Y\)\}}} \, \frac{1}{2} \q'_+ (\s (\l)) \vd L^{\s (\l)}_t (U) = 0.
\end{align}
Now, if \(\l\) is absorbing for \(Y\), we deduce that \(r \l = 0\), because \(\P (T_{\s (\l)} (U) < T / 2) > 0\) by \cite[Theorem~1.1]{BrugRuf16}. In case \(\l\) is reflecting for \(Y\), 
\cite[Theorem~136, p.~160]{freedman} yields that \(\P\)-a.s.
		\begin{align*} 
		- r \alpha \1_{\{U_t = \s (\l)\}} \vd t = - r \l \m^U (\{\s (\l)\}) \vd L^{\s (\l)}_t (U),
		\end{align*}
and consequently, with \eqref{eq: IMPR - pf} we get \(\P\)-a.s.
\begin{align} \label{eq: boundary identity}
\Big( \frac{1}{2} \q'_+ (\s (\l)) - r \alpha \m^U (\{\s (\alpha)\}) \Big) \vd L^{\s (\l)}_t (U) = 0.
\end{align}
Because \(U = \s (\Y)\) is a diffusion on natural scale, (still in the case where \(\l\) is reflecting) we deduce from \cite[Lemma~C.18]{CU22} that \(\P\)-a.s. \(L^{\s (\l)}_t (\U) > 0\) for all \(t > T_{\s (\l)} (\U)\).
Now, as \(\P (T_{\s (l)} (\U) < T / 2) > 0\) by \cite[Theorem~1.1]{BrugRuf16}, \eqref{eq: boundary identity} yields \(\frac{1}{2} \q'_+ (\s (\l)) - r \alpha \m^U (\{\s (\alpha)\}) = 0\). In summary, the properties from~(i) hold.

		\smallskip
		{\em Proof of (ii) + (iii):}
		In case \(\l\) is an absorbing boundary point, recall that we have \(r = 0\) or \(\l = 0\) by part (i). Hence, using \cite[Theorem~V.49.1]{RogWilV2} for absorbing \(\l\) and \cite[Theorem~136, p.~160]{freedman} for reflecting or inaccessible \(\l\), we obtain that \(\P\)-a.s.
		\begin{align} \label{eq: ds integral}
			\int_0^t r \q (U_s) \vd s = \int_{\s (J)} r \q (x) L^x_t (U) \, \m^{\U} (\rd x), \quad t \in [0, T].
		\end{align}
        Take a set \(B \in \mathcal{B} (\s (J^\circ))\).
        Recalling Lemma~\ref{lem: semimartingale decomposition concrete}, \eqref{eq: IMPR eq} and \eqref{eq: ds integral}, and using the semimartingale occupation time formula (\cite[Theorem~IV.45.1]{RogWilV2}), we obtain that \(\P\)-a.s. 
        \begin{align} \label{eq: main idenity (ii), (iii)}
     - \int_{B}  \, r \q (x) \vd L^x_t (U) \, \m^{\U} (\rd x) + \frac{1}{2} \int_{B} \vd L^x_t (U) \, \q'' (dx) = e^{- rt} \, \theta_t \, \int_B \, \big[ \q'_+ (x) \big]^2 \vd L^x_t (U) \vd x.
        \end{align} 
		Let \(G \in \mathcal{B} (\s (J^\circ))\) be a Lebesgue null set such that 
		\[
		\m^{\U}_\si (G \cap \, \cdot \, ) = \m^{\U}_\si, \quad \q''_\si (G \cap \, \cdot \, ) = \q''_\si.
		\]
		Such a set exists by the definition of the singular parts.
        Taking \(B = H \cap G\), with \(H \in \mathcal{B} (\s (J^\circ))\), in \eqref{eq: main idenity (ii), (iii)} yields that \(\P\)-a.s.
        \begin{align} \label{eq: main for (ii)}
         \int_H
         r \q (x) \vd L^x_t (U) \, \m^{\U}_\si (\rd x)
         =
         \int_H
         \frac{1}{2} \vd L^x_t (U) \, \q''_\si (\rd x).
        \end{align}
        Similarly, taking \(B = H \cap G^c \cap \{\q'_+ = 0\}\) in \eqref{eq: main idenity (ii), (iii)} gives \(\P\)-a.s.
         \begin{align} \label{eq: main for iii}
         \int_H
         \1_{\{\q'_+ (x) = 0\}}  \, r \q (x) \vd L^x_t (U) \fmu (x) \vd x
         =
         \int_H
         \frac{1}{2} \1_{\{\q'_+ (x) = 0\}} \vd L^x_t (U)\q'' (x) \vd x.
        \end{align}
		Take \(z_0 \in \s (J^\circ)\) with \(z_0 \not = \s (x_0)\). 
		To ease our notation, we write 
		\[
		\U^{\st} \triangleq (U_{t \wedge T_{\s(\l)} (U)})_{t \in [0, T]},
		\]
		which is the diffusion \(U\) stopped at its accessible boundary point.
		By \cite[Theorem 1.1]{BrugRuf16}, 
		\[
		\P (T_{z_0} (\U^{\st}) < T ) > 0.
		\]
		Further, as \(\U^{\st}\) is a local martingale (\cite[Corollary~V.46.15]{RogWilV2}), recall from \cite[Corollary~29.18]{Kal21} that \(\P\)-a.s. 
		\begin{align} \label{eq: local time support}
			\Big\{ L^x_{t}(\U^{\st}) > 0 \Big\} = \Big\{ \min_{s \in [0, t]}\U^{\st}_s < x < \max_{s \in [0, t]}\U^{\st}_s \Big\}, \quad x \in \bR, \, t \in [0, T].
		\end{align} 
		Consequently, \(\P\)-a.s. 
		\[
		\Big\{T_{z_0} (\U^{\st}) < T, \, L_T^x (\U^{\st}) > 0 \text{ for all } x \in (\s(x_0) \wedge z_0, \s(x_0) \vee z_0) \Big\} = \Big\{ T_{z_0} (\U^{\st}) < T \Big\},
		\]
		which is a set of positive probability. Now, using that \(\P\)-a.s. \(L^\cdot_{\cdot \, \wedge T_{\s(\l)} (U)} (U) = L^\cdot_\cdot (U^{\st})\) and  integrating 
		\[ 
		(t, x) \mapsto \frac{\1_{[0, T_{\s (\l)} (U) \wedge T]}(t) \, \1_{F} (x)}{L^x_{T_{\s (\l)} (U) \wedge T} (U)} \, \1_{\{L^x_{T_{\s (\l)} (U)} \, >\,  0\}}, \quad F \in \mathcal{B}((\s(x_0) \wedge z_0, \s(x_0) \vee z_0)), 
		\] 
		against both sides in \eqref{eq: main for (ii)} and~\eqref{eq: main for iii}
        yields that on \(\mathcal{B}((\s(x_0) \wedge z_0, \s(x_0) \vee z_0))\)
        \begin{align*} 
        r \q (x) \, \m^U_\si (\rd x) &= \frac{1}{2} \, \q''_\si (\rd x), \\
r \q (x) \fmu (x) \1_{\{\q'_+ (x) = 0\}} \vd x &= \frac{1}{2} \,\q'' (x) \1_{\{\q'_+ (x) = 0\}} \vd x
        \end{align*}
(here we consider both sides in \eqref{eq: main for (ii)} and~\eqref{eq: main for iii}
as measures on $[0, T] \times \bR$).
As \(z_0 \in \s (J^\circ) \setminus \{\s (x_0)\}\) was arbitrary,
we conclude that (ii) and~(iii) hold.
	\end{proof}

	\begin{proof}[Proof of Theorem~\ref{theo: main_NIP} (Sufficiency)] Assume that the properties (i)--(iii) hold. Again, to ease our presentation, we only consider the case \(J = [\l, \infty)\) with reflecting \(\l\). The other cases are proved by similar methods. 
		
		The equation \eqref{eq: main idenity (ii), (iii)} gives the impression that the process \(\theta = (\theta_t)_{t \in [0, T]}\) from \eqref{eq:200125a1} is a good candidate for an IMPR. 
		We now verify this impression.
		Using Lemma~\ref{lem: semimartingale decomposition concrete}, \eqref{eq: ds integral} and the assumptions (i) and (ii), we get that \(\P\)-a.s.
		\begin{align*}
			e^{rt} \vd A_t &= - r \q (U_t) \vd t+ \frac{1}{2} \q'_+ (\s(\l)) \vd L^{\s (\l)}_t (U) + \frac{1}{2} \int_{\s (J^\circ)} \vd L^x_t (U) \, \q'' (\rd x) 
            \\&= -\int_{\s (J)} r \q (x)\vd L^x_t (\U) \, \m^U (\rd x) + r \q (\s (\l)) \m^U (\{ \s (\l)\}) \vd L^{\s (\l)}_t (U) 
            \\&\hspace{1.5cm}+ \frac{1}{2} \int_{\s (J^\circ)} \vd L^x_t (U) \,  \q'' (\rd x) 
			\\&= - \int_{\s (J^\circ)} r \q (x) \vd L^x_t (U) \, \m^{\U} (\rd x) + \frac{1}{2} \int_{\s (J^\circ)} \vd L^x_t (U) \, \q'' (\rd x)
			\\&=  - \int_{\s (J^\circ)} r \q (x) \vd L^x_t (U) \, \fmu (x) \vd x + \frac{1}{2} \int_{\s (J^\circ)} \vd L^x_t (U) \q'' (x) \vd x
			\\&\hspace{1.5cm} + \frac{1}{2} \int_{\s (J^\circ)} \vd L^x_t (U) \big(  \q''_\si (\rd x) - 2r \q (x) \m^{\U}_\si (\rd x) \big) 
			\\&=  \int_{\s (J^\circ)} \Big( - r \q (x) \fmu (x) + \frac{1}{2}\q'' (x)\Big) \vd L^x_t (U) \vd x.
		\end{align*}
		Using again Lemma~\ref{lem: semimartingale decomposition concrete}, the semimartingale occupation time formula (\cite[Theorem~IV.45.1]{RogWilV2}) and assumption (iii), we obtain \(\P\)-a.s. 
		\begin{align*}
			\theta_t\, \rd \langle M, M\rangle_t &=  \frac{e^{rt} (- r \q (U_t) \fmu (U_t) + \frac{1}{2}\q'' (U_t))}{[\q'_+ (U_t)]^2} \1_{\{\q'_+ (U_t) \not = 0\}} e^{- 2rt} \big[ \q'_+ (U_t) \big]^2 \vd \langle U, U\rangle_t \phantom \int
			\\&= e^{- rt} \Big( - r \q (U_t) \fmu (U_t) + \frac{1}{2}\q'' (U_t) \Big) \1_{\{\q'_+ (U_t) \not = 0\}} \vd \langle U, U\rangle_t \phantom \int
			\\&= e^{- rt} \int_{\s (J^\circ)} \Big( - r \q (x) \fmu (x) + \frac{1}{2}\q'' (x) \Big) \1_{\{\q'_+ (x) \not = 0\}} \vd L^x_t (U) \vd x
			\\&= e^{- rt} \int_{\s (J^\circ)} \Big( - r \q (x) \fmu (x) + \frac{1}{2}\q'' (x) \Big) \vd L^x_t (U) \vd x.
		\end{align*}
        In summary, we have \(\theta_t \vd \langle M, M\rangle_t = \rd A_t\) and 
		Theorem~\ref{theo: FTAP_NIP} implies that the NIP condition holds. Furthermore, we proved that the process from \eqref{eq:200125a1} defines an IMPR.
	\end{proof}

It is worth discussing what Theorem~\ref{theo: main_NIP} yields in the zero interest rate regime.

    \begin{corollary} \label{coro: NIP r = 0} 
Assume that \(r = 0\).
Then NIP holds if and only if the following two conditions are satisfied:
		\begin{enumerate}
			\item[\textup{(i)}]
            Every reflecting boundary point $b$ of \(\Y\) satisfies \(\q' ( \s (b)) = 0\), where \(\q'\) is the right- or left-hand derivative of \(\q\) depending on whether \(b\) is the left or right boundary point.
			\item[\textup{(ii)}]
            The inverse scale function $\q$ is a $C^1$-function on $\s(J^\circ)$ with a (locally) absolutely continuous derivative.
		\end{enumerate}
    \end{corollary}

In part~(ii) of Corollary~\ref{coro: NIP r = 0} and in similar contexts below,
the expression ``(locally) absolutely continuous on $\s(J^\circ)$''
means absolutely continuous on every compact subinterval of $\s(J^\circ)$. The proof of Corollary~\ref{coro: NIP r = 0} is given after the following discussion.

\begin{discussion}\label{disc:040225a1}
(i) 
As we illustrate in Example~\ref{ex: bessel} below, part~(i) in Corollary~\ref{coro: NIP r = 0} is a non-empty condition in the sense that there are examples of models with reflecting boundaries that satisfy NIP. This observation is surprising at first glance, because there is already some evidence in the literature (see \cite{BDH23, CU23}) that reflecting boundaries lead to arbitrage opportunities even in rather
strong
forms.

(ii)
We also observe that part~(ii) in Corollary~\ref{coro: NIP r = 0} excludes the case where $Y$ is a \emph{skew diffusion}.
Naturally generalizing the concept of the skew Brownian motion,
we understand the \emph{skew diffusion} as a general diffusion with the scale function $\s$ that has a kink at some point $c\in J^\circ$ in the sense that the left- and right-hand derivatives $\s'_-(c)$ and $\s'_+(c)$ exist and belong to $(0,\infty)$ but $\s'_-(c)\ne\s'_+(c)$.
On the other hand, part~(ii) of Corollary~\ref{coro: NIP r = 0} does not exclude irregularities of the type when $\s'(c)=\infty$ at some point $c\in J^\circ$
(see Examples \ref{ex: x^3} and~\ref{ex:060225a1} below).

(iii)
It is interesting to discuss if there is a relation between the no-arbitrage notions in the zero interest rate regime and the \emph{representation property} (RP) of $Y$, which is often connected to market completeness (cf. \cite[Section VII.2.d]{shir}).
We first recall what the RP means.
Let $\bfFY$ denote the right-continuous filtration generated by $Y$.
By Stricker's theorem (see \cite[Theorem~9.19]{jacod79}), the $\mathbf F$-$\P$-semimartingale $Y$ is also an $\bfFY$-$\P$-semimartingale.
We say that the \emph{RP} holds for $Y$ if every $\bfFY$-$\P$-local martingale $M$ has a representation
$$
M=M_0+\int_0^\cdot H_s\vd Y_s^c,
$$
where $Y^c$ denotes the continuous local martingale part of $Y$ w.r.t. the filtration $\bfFY$ and $H$ is an $\bfFY$-predictable process integrable w.r.t. $Y^c$.
The main result of \cite{CU24} states that the following are equivalent:
\begin{itemize}
\item
the RP holds for $Y$,

\item
$\llambda(\{x\in\s(J^\circ):\q'(x)=0\})=0$,

\item
the scale function $\s$ is absolutely continuous on all compact subintervals of $J^\circ$.
\end{itemize}
In Example~\ref{ex:060225a1} below we have NIP but $\llambda(\{x\in\s(J^\circ):\q'(x)=0\})>0$, that is, the RP fails for~$Y$.
Thus, for $r=0$, NIP does not imply the RP.
Below we will see that this is different for the stronger no-arbitrage notions.

(iv)
A comparison of Theorem~\ref{theo: main_NIP} and Corollary~\ref{coro: NIP r = 0}
also reveals the interesting and surprising new effect that in the non-zero interest rate regime the NIP condition requires less regularity of the scale function than in the zero interest rate regime.
Indeed, by Corollary~\ref{coro: NIP r = 0}, in the case $r=0$, NIP does not allow the second derivative measure \(\q'' (\rd x)\) to have a singular part.
On the other hand, Theorem~\ref{theo: main_NIP} shows that, in the case $r\ne0$, the measure \(\q'' (\rd x)\) could have a singular part (coming for instance from skewness), as long as it is compensated by a singular part of the speed measure (coming for instance from stickiness).
This effect appears to be new and it is only visible in the presence of non-zero interest rates.
Example~\ref{ex: sticky skew} below illustrates how skewness and stickiness can cancel each other.
\end{discussion}

For the proof of Corollary~\ref{coro: NIP r = 0} we need the following elementary observation.

\begin{lemma}\label{lem:270125a1}
Consider an open interval $I\subset\bR$ and a function $f\colon I\to\bR$ such that
\begin{itemize}
\item[\textup{(a)}]
$f'$ exists and is finite on $I$,

\item[\textup{(b)}]
$f''$ exists and is finite $\llambda$-a.e. on $I$.
\end{itemize}
Then, $f''\1_{\{f'=0\}}=0$ $\llambda$-a.e. on $I$.
\end{lemma}

\begin{proof}
Take $R\in\mathcal B(I)$ such that $f''$ exists and is finite on $R$ and $\llambda(I\setminus R)=0$.
We set $A\triangleq\{x\in I\colon f'(x)=0\}$ and consider the decomposition
\begin{equation}\label{eq:270125a1}
A=A^{\on{int}}\sqcup A^{\on{acc}}\sqcup A^{\on{iso}},
\end{equation}
where $A^{\on{int}}$ (resp., $A^{\on{acc}}$; resp., $A^{\on{iso}}$)
is the subset of $A$ containing all interior points (resp., accumulation points; resp., isolated points) of $A$.
Notice that $A^{\on{iso}}$ is at most countable hence does not matter for the claim of the lemma.
Furthermore,
everywhere on $A^{\on{int}}$ we clearly have $f''=0$,
while on $A^{\on{acc}}\cap R$ we have
either $\rd^+ f'/\rd x=0$ or $\rd^- f'/\rd x=0$
(use the property of the accumulation points),
which again yields $f''=0$ everywhere on $A^{\on{acc}}\cap R$.
Thus, $\llambda$-a.e. on $A$ we have $f''=0$, as needed.
\end{proof}

\begin{proof}[Proof of Corollary~\ref{coro: NIP r = 0}]
We first remark that (ii) in the statement of Corollary~\ref{coro: NIP r = 0}
is nothing else than an equivalent reformulation of $\q''_\si(\rd x)=0$ on $\mathcal B(\s(J^\circ))$.
From Theorem~\ref{theo: main_NIP}, we obtain that,
in the case $r=0$,
NIP is equivalent to (i)--(iii), where (i) and~(ii) are as in the formulation of Corollary~\ref{coro: NIP r = 0} and (iii) is as follows:
\begin{enumerate}
\item[\textup{(iii)}]
$\llambda$-a.e. on $\s(J^\circ)$ we have \(\q''\1_{\{\q' = 0\}} = 0\).
\end{enumerate}
It remains only to observe that (ii) implies~(iii).
Indeed, this is just an application of Lemma~\ref{lem:270125a1},
where (ii) guarantees that the assumptions of Lemma~\ref{lem:270125a1} are satisfied for $f=\q$ and $I=\s(J^\circ)$.
The proof is complete.
\end{proof}

\begin{discussion}
As a little technical observation, in the context of Corollary~\ref{coro: NIP r = 0}, we notice that the decomposition~\eqref{eq:270125a1} for the set
$$
A\triangleq\{x\in\s(J^\circ):\q'(x)=0\}
$$
has more structure than in the more general situation of Lemma~\ref{lem:270125a1}. Namely,
\begin{itemize}
\item
as $\q'$ is continuous, $A$ is closed in $\s(J^\circ)$;

\item
as $\q$ is strictly monotone, $A^{\on{int}}=\emptyset$.
\end{itemize}
Now, it is natural to ask whether $A^{\on{iso}}$ and/or $A^{\on{acc}}$ can be non-empty in the context of Corollary~\ref{coro: NIP r = 0}.
The answer is affirmative.
In Example~\ref{ex: x^3} below
we have $\s(J^\circ)=\bR$ and $\q(x)=x^3$. Hence, the set $A$ is just one point, so $A^{\on{iso}}\ne\emptyset$.
Furthermore, Example~\ref{ex:060225a1} illustrates that the set $A$ can have positive Lebesgue measure.
In view of the fact that $A^{\on{int}}=\emptyset$, in that example $A^{\on{acc}}\ne\emptyset$ (and even $\llambda(A^{\on{acc}})>0$).
\end{discussion}

    \section{A Deterministic Characterization of NSA}
    \label{sec_NSA}

   We now establish a deterministic characterization of the NSA condition that only depends on scale and speed. 
  The following is the main result of this section. Below we also provide a simplified version for the zero interest rate regime. 

    	\begin{theorem}\label{theo: main_NSA}
		NSA holds if and only if
        \textup{(i)}--\textup{(iii)} from Theorem~\ref{theo: main_NIP} hold and additionally:
		\begin{enumerate}
			\item[\textup{(iv)}]
            The function \(\varphi \colon \s (J^\circ) \to \bR\) defined by 
			\begin{align*}
				\varphi \triangleq \frac{\frac{1}{2}\q'' - r \q \fmu}{\q'} \1_{\{\q' \ne 0\}},\quad\text{on } \s (J^\circ),
			\end{align*}
			satisfies 
			\[
			\varphi \in L^2_\textup{loc} (\s (J^\circ)),
			\]
            and for every accessible boundary $b\in J\setminus J^\circ$ it holds:
            $$
            \text{If }b\text{ is reflecting, then }
            \int_{I} \varphi^2 (x) \vd x < \infty
            $$
            for some (equivalently, for every) open interval \(I \subsetneq \s (J^\circ)\) with \(\s (b)\) as endpoint.
		\end{enumerate}
	\end{theorem}

    By virtue of Theorem~\ref{theo: main_NIP}, (iv) is precisely what distinguishes NIP and NSA, i.e., what is needed in addition to NIP so that NSA holds.

\begin{proof}
By Theorem~\ref{theo: FTAP_NSA}, NSA holds if and only if NIP and \eqref{eq:200125a2} hold. Thanks to Theorem~\ref{theo: main_NIP}, NIP is equivalent to (i)--(iii). Hence, Theorem~\ref{theo: main_NSA} follows once we prove that, under (i)--(iii), \eqref{eq:200125a2} is equivalent to (iv). 
To streamline our presentation, we only consider the case $J^\circ=(\alpha,\infty)$, stressing that all other cases can be treated by similar methods. For the remainder of this proof, we assume that (i)--(iii) hold.

\smallskip 
As a preparation, we deduce from \eqref{eq:200125a1}, \eqref{eq: final DMD} and~\eqref{eq: final DMD - re} that,
if $\alpha$ is absorbing, then,
$\P$-a.s. for $s< T_\alpha(Y)\;(=T_{\s(\alpha)}(U))$, 
\begin{equation}\label{eq:310125a1}
\begin{split}
\theta_s^2\vd\langle M,M\rangle_s 
&= e^{2 r s } \gamma^2 (U_s) \, e^{- 2 r s} \big[ \q'_+ (U_s)\big]^2 \vd \langle U, U \rangle_s
\\
&=
\gamma^2(U_s)\big[\q'_+(U_s)\big]^2\rd\langle U,U\rangle_s,
\end{split}
\end{equation} 
while, if $\alpha$ is inaccessible or reflecting, the above identity holds globally (i.e., \(\P\)-a.s. for all \(s \in [0, T]\)).
As $\llambda$-a.e. $\gamma\q'_+=\varphi$ on $\s(J^\circ)$,
the semimartingale occupation time formula yields that
$\P$-a.s. for all $t\in[0,T_\alpha(Y)\wedge T)$, if \(\l\) is absorbing, or for all $t\in[0,T]$, if \(\l\) is inaccessible or reflecting,
\begin{equation}\label{eq:200125a3}
\int_0^t \theta_s^2\vd\langle M,M\rangle_s
=
\int_0^t \varphi^2(U_s)\vd\langle U,U\rangle_s
=
\int_{\s(J^\circ)}\varphi^2(x)L_t^x(\U)\vd x.
\end{equation}
(In fact,
both \eqref{eq:310125a1} and~\eqref{eq:200125a3} hold globally even when $\alpha$ is absorbing, but this fact is not needed at this point.)

\smallskip 
We now turn to the main body of the proof, starting with the case where (iv) holds.
Let $\alpha$ be inaccessible or reflecting.
Then \eqref{eq:200125a3} implies that,
$\P$-a.s. for all $t\in[0,T]$,
$\int_0^t \theta_s^2 \vd\langle M,M\rangle_s<\infty$
(notice that $\P$-a.s. the function $\s(J)\ni x\mapsto L_t^x(U)$ is bounded as c\`adl\`ag function with a compact support in $\s(J)$).
It follows that \eqref{eq:310125a2}, hence also~\eqref{eq:200125a2}, is satisfied.
Now, consider the case of an absorbing $\alpha$.
As above, in this case, \eqref{eq:200125a3} yields that,
$\P$-a.s. for all $t\in[0,T_\alpha(Y)\wedge T)$,
$\int_0^t \theta_s^2 \vd\langle M,M\rangle_s<\infty$.
As $M$ is constant on $[T_\alpha(Y)\wedge T,T]$,
the jump to infinity as described by the left-hand side of~\eqref{eq:200125a2}
cannot happen after $T_\alpha(Y)\wedge T$.
Thus, \eqref{eq:200125a2} is satisfied again.
We thus proved that, if (iv) is satisfied, then \eqref{eq:200125a2} holds.

\smallskip 
Lastly, we assume that (iv) is violated and we prove that \eqref{eq:200125a2} fails, too.
There are two possibilities: 
\begin{enumerate}
\item[(a)]
either there is a point $u\in(\s(\alpha),\infty)$ such that
\begin{equation}\label{eq:200125a4}
\int_I \varphi^2(x)\vd x=\infty
\end{equation}
for every open neighborhood $I\subset(\s(\alpha),\infty)$ of $u$,

\item[(b)]
or $\alpha$ is reflecting and \eqref{eq:200125a4} holds for every open interval $I\subset(\s(\alpha),\infty)$ with $\s(\alpha)$ as endpoint.
In this latter case we set $u\triangleq\s(\alpha)$.
\end{enumerate}
Using the semimartingale occupation time formula twice, 
we obtain that, $\P$-a.s for all $h\in(0,\infty)$ and for all
\emph{bounded} Borel functions $f\colon\s(J)\to\bR$,
\begin{equation}\label{eq:220125a1}
\int_{T_u(U)}^{T_u(U)+h}
f(U_s)\vd\langle U,U\rangle_s
=
\int_{\s(J^\circ)}
f(x)\big(
L_{T_u(U)+h}^x(U)
-
L_{T_u(U)}^x(U)
\big)\vd x
\end{equation}
(as $\llambda(\s(J)\setminus\s(J^\circ))=0$, we can integrate over $\s(J^\circ)$ on the right-hand side).
By a standard monotone class argument, \eqref{eq:220125a1} extends to all \emph{nonnegative} Borel functions $f\colon\s(J)\to[0,\infty]$
(we cannot do the subtraction in general, as we may get $\infty-\infty$).
Recall that, if $\alpha$ is absorbing, then $\P$-a.s. for $s<T_\alpha(Y)$ we have
$\theta_s^2\vd\langle M,M\rangle_s=\varphi^2(U_s)\vd\langle U,U\rangle_s$,
while, if $\alpha$ is inaccessible or reflecting, this holds globally.
Take some $h_0\in(0,T)$.
By \cite[Theorem 1.1]{BrugRuf16}, it holds that $\P(T_u(U)<T-h_0)>0$.
We further notice that, if $\alpha$ is absorbing, then
$\{T_u(U)<T-h_0\}\subset\{T_u(U)<T_\alpha(Y)\}$.
The above discussions and \eqref{eq:220125a1} with $f=\varphi^2$ yield that
$\P$-a.s. on the event $\{T_u(U)<T-h_0\}$ of strictly positive $\P$-probability for all \emph{sufficiently small}\footnote{More precisely,
``sufficiently small'' means in this context that
$h<h_0$ whenever $\alpha$ is inaccessible or reflecting,
while $h<h_0\wedge(T_\alpha(Y)-T_u(U))$ whenever $\alpha$ is absorbing.}
$h>0$ it holds
\begin{equation}\label{eq:220125a2}
\int_{T_u(U)}^{T_u(U)+h}
\theta_s^2\vd\langle M,M\rangle_s
=
\int_{\s(J^\circ)}
\varphi^2(x)\big(
L_{T_u(U)+h}^x(U)
-
L_{T_u(U)}^x(U)
\big)\vd x.
\end{equation}
Observe that $\P$-a.s. for all $h>0$ the function
$x\mapsto L_{T_u(U)+h}^x(U)-L_{T_u(U)}^x(U)$
is continuous on $\s(J^\circ)$
(\cite[Theorem V.49.1]{RogWilV2})
and right-continuous at the point $\s(\alpha)$
(recall that we work with the right semimartingale local time in this paper).
Furthermore, by the strong Markov property of $U$ together with \cite[Lemma~C.18]{CU22},
$\P$-a.s. for all $h>0$ we have
$L_{T_u(U)+h}^u(U)-L_{T_u(U)}^u(U)>0$.
Now, the choice of the point $u$ together with \eqref{eq:220125a2} imply that \eqref{eq:200125a2} is violated.
This concludes the proof.
\end{proof}

\begin{remark}\label{rem: NSA NUPBR proof}
The proof of Theorem~\ref{theo: main_NSA} shows that if NIP holds and all boundaries of \(\Y\) are either inaccessible or reflecting, then \eqref{eq:200125a2} already implies~\eqref{eq:310125a2}.
As a consequence, in the absence of absorbing boundary points, NSA and NUPBR are equivalent. In contrast, if one of the boundary points is absorbing, it is possible that \eqref{eq:200125a2} holds while \[\int_0^T \theta^2_s \vd \langle M, M\rangle_s = \infty\] with positive \(\P\)-probability
(see the proof of Corollary~\ref{cor:070225a1} below for an explicit example in our setting).
\end{remark} 

For the zero interest rate regime, we deduce the following collection of interesting equivalent characterizations of NSA.

\begin{corollary} \label{coro: NSA r = 0}
For $r=0$ the following are equivalent:
\begin{enumerate}
\item[\textup{(i)}]
NSA holds.

\item[\textup{(ii)}]
The scale function $\s$ is a $C^1$-function on $J^\circ$ with a strictly positive (locally) absolutely continuous derivative,
$\s''/\s'\in L^2_\textup{loc}(J^\circ)$ and the boundary points of $J$ are inaccessible or absorbing for~$Y$.

\item[\textup{(iii)}]
The scale function $\s$ admits the representation
\begin{equation}\label{eq:220125a3}
\s(x)=\int^x\exp\Big\{\int^y \mu(z)\vd z\Big\}\vd y,
\quad x\in J^\circ,
\end{equation}
with some $\mu\in L^2_\textup{loc}(J^\circ)$ and the boundary points of $J$ are inaccessible or absorbing for~$Y$.

\item[\textup{(iv)}]
The inverse scale function $\q$ is a $C^1$-function on $\s(J^\circ)$ with a strictly positive (locally) absolutely continuous derivative,
$\q''/\q'\in L^2_\textup{loc}(\s(J^\circ))$ and the boundary points of $J$ are inaccessible or absorbing for~$Y$.

\item[\textup{(v)}]
The inverse scale function $\q$ admits the representation
\begin{equation}\label{eq:220125a4}
\q(x)=\int^x\exp\Big\{\int^y \nu(z)\vd z\Big\}\vd y,
\quad x\in\s(J^\circ),
\end{equation}
with some $\nu\in L^2_\textup{loc}(\s(J^\circ))$ and the boundary points of $J$ are inaccessible or absorbing for~$Y$.
\end{enumerate}
\end{corollary}

\begin{proof}
Recasting (i)--(iii) from Theorem~\ref{theo: main_NIP} and (iv) from Theorem~\ref{theo: main_NSA} for the case $r=0$, we obtain that NSA holds if and only if the following conditions (a)--(d) are satisfied\footnote{As proved in Corollary~\ref{coro: NIP r = 0}, condition~(c) can be dropped here. But for this proof it is convenient to have (c) explicitly.}:
\begin{enumerate}
\item[(a)]
Every reflecting boundary point $b$ of \(\Y\) satisfies \(\q' ( \s (b)) = 0\), where \(\q'\) is the right- or left-hand derivative of \(\q\) depending on whether \(b\) is the left or right boundary point.

\item[(b)]
\(\q''(\rd x) = \q'' (x) \vd x\) on $\mathcal B(\s(J^\circ))$.

\item[(c)]
$\llambda$-a.e. on $\s(J^\circ)$ we have \(\q''\1_{\{\q' = 0\}} = 0\).

\item[(d)]
It holds
\begin{equation}\label{eq:270125a2}
\frac{\q''}{\q'}\1_{\{\q'\ne0\}}\in L^2_\textup{loc}(\s(J^\circ))
\end{equation}
and, for every reflecting boundary point $b$, we have
\begin{equation}\label{eq:270125a3}
\int_I
\Big(\frac{\q''}{\q'}\Big)^2(x)\,\1_{\{\q'\ne0\}}(x)\vd x<\infty
\end{equation}
for some (equivalently, for every) open interval $I\subsetneq\s(J^\circ)$ with $\s(b)$ as endpoint.
\end{enumerate}
It is obvious that (v) implies (a)--(d) and hence, (v)$\implies$(i).
Next, we prove the reverse implication (i)$\implies$(v).
Assume that (a)--(d) hold.
It follows from (b), (c) and~\eqref{eq:270125a2} that
$\q$ is a $C^1$-function on $\s(J^\circ)$ with a (locally) absolutely continuous derivative $\q'$ such that there exists a Borel function $\nu\in L^2_\mathrm{loc}(\s(J^\circ))$
with
$$
\q'' = \nu \q'\quad\llambda\text{-a.e. on }\s(J^\circ)
$$
(we can simply define $\nu$ as the function in~\eqref{eq:270125a2}).
This is a differential equation for the function $\q'$ in the sense of Carath\'eodory.
By a standard existence and uniqueness result for such equations (see \cite[Theorem~XVIII, p. 121]{walter}), the above differential equation (in the sense of Carath\'eodory) has a unique (up to multiplicative constants) solution
\begin{equation}\label{eq:270125a4}
\q' = \exp \Big\{\int^\cdot \nu (z) \vd z \Big\} \quad \text{on } \s (J^\circ).
\end{equation}
It remains only to prove that there are no reflecting boundaries.
For contradiction, assume that $b$ is a reflecting boundary point.
Take an open interval $I\subsetneq\s(J^\circ)$ with $\s(b)$ as endpoint.
As $b$ is accessible, $|\s(b)|<\infty$ and hence, $\llambda(I)<\infty$.
Therefore, by~\eqref{eq:270125a3}, we get that
$$
\nu\in L^2(I)\subset L^1(I),
$$
which, together with~\eqref{eq:270125a4}, yields $\q'(\s(b))>0$.
This contradiction to~(a) proves that there are no reflecting boundaries.
We thus established the equivalence (i) $\Longleftrightarrow$ (v).

Further, the equivalences (ii) $\Longleftrightarrow$ (iii) and (iv) $\Longleftrightarrow$ (v) are straightforward.
We only notice that, in the implications (ii)$\implies$(iii) and (iv)$\implies$(v),
the representations \eqref{eq:220125a3} and~\eqref{eq:220125a4} are again due to the mentioned uniqueness result for differential equations in the sense of Carath\'eodory.

To prove the implication (ii)$\implies$(iv), we notice that, under~(ii),
$\q$ is a $C^1$-function on $\s(J^\circ)$ such that
$\q'(x)=1/\s'(\q(x))$ for $x\in\s(J^\circ)$.
Being $C^1$, the function $\q$ is (locally) absolutely continuous.
This fact, together with the above formula for $\q'$ and the (local) absolute continuity and strict positivity of $\s'$, yields that $\q'$ is (locally) absolutely continuous on $\s(J^\circ)$.
Further, by a straightforward calculation,
\begin{equation}\label{eq:050225a3}
\frac{\q''}{\q'}(x)=-\frac{\s''}{\s'}(\q(x))\,\q'(x),
\quad\text{for }\llambda\text{-a.a. }x\in\s(J^\circ).
\end{equation}
Finally, the latter formula,
the local square integrability of $\s''/\s'$ on $J^\circ$
and the local boundedness of $\q'$ on $\s(J^\circ)$
imply the local square integrability of $\q''/\q'$ on $\s(J^\circ)$.
The reverse implication (iv)$\implies$(ii) is symmetric.
Thus, we established the remaining equivalence (ii) $\Longleftrightarrow$ (iv).
This concludes the proof.
\end{proof}

\begin{discussion}\label{disc:040225a2}
(i)
A comparison of Theorem~\ref{theo: main_NSA} and Corollary~\ref{coro: NSA r = 0}
reveals that NSA forces accessible boundaries for $Y$ to be absorbing only in the case $r=0$.
In contrast, it is possible for $r\ne0$ that NSA (and even the stronger condition NUPBR) holds although \(\Y\) has reflecting boundaries, see Example~\ref{ex: NUBPR ref interest rate} below.

(ii)
In contrast to Discusion~\ref{disc:040225a1}~(ii), in the case $r=0$, NSA excludes not only skew diffusions but also the possibility of $\s'(c)=\infty$ at some point $c\in J^\circ$.

(iii)
In contrast to Discussion~\ref{disc:040225a1}~(iii), for $r=0$, NSA (hence, also NUPBR) implies the RP for $Y$.
This is a direct consequence of Corollary~\ref{coro: NSA r = 0} and the characterization of the RP recalled in Discussion~\ref{disc:040225a1}~(iii).

(iv)
As in the case of NIP, we again see that in the non-zero interest rate regime the NSA condition requires less regularity of the scale function than in the zero interest rate regime
(cf. Discussion~\ref{disc:040225a1}~(iv)).

(v)
The fact that, in the case $r=0$, NSA forces accessible boundaries to be absorbing is established in the proof of Corollary~\ref{coro: NSA r = 0} in a purely analytic way.
It is instructive to deduce this directly from the definition of NSA.
Assume for contradiction that the left boundary point $\alpha$ is reflecting and consider the strategy $H\triangleq\1_{(T_\alpha(Y)\wedge T,T]}$.
As $r=0$, we have $S=Y$ and hence,
$$
\int_0^t H_s\vd S_s=Y_t-Y_{t\wedge T_\alpha(Y)}\ge0 ,\quad t\in[0,T],
$$
which means that $H\in\mathcal A_0$.
It is also reasonable to expect that
$$
\P\Big(\int_0^T H_s\vd S_s>0\Big)
=
\P(Y_T-Y_{T\wedge T_\alpha(Y)}>0)
>0.
$$
We refer to the proof of Lemma~5.7 in \cite{CU23} for the formal argument.
Thus, $H$ generates a strong arbitrage opportunity.
\end{discussion}

	\section{A Deterministic Characterization of NUPBR}
	\label{sec_NUPBR}
	
	For our setting \emph{without} interest rate, a deterministic characterization of the NUPBR condition has been established in the paper \cite{CU23}. The proofs there used deep results on separating times for general diffusions that were established in \cite{CU22}.
    The present paper suggests, in particular, a different way of proving the characterization of NUPBR.
    But it is worth noting that, with our current methods, we cannot handle the other no-arbitrage notions considered in \cite{CU23}
    (namely, NA and NFLVR),
    while\cite{CU23} neither considered NSA nor NIP
    and we do not treat the possibility of a non-zero interest rate
    (the latter also means that many of the surprising effects discussed in this paper are not within the scope of \cite{CU23}).

\smallskip 
    The main result of this section is as follows.

	\begin{theorem}\label{theo: main_NUPBR}
		NUPBR holds if and only if
        \textup{(i)}--\textup{(iii)} from Theorem~\ref{theo: main_NIP}
        and \textup{(iv)} from Theorem~\ref{theo: main_NSA}
        hold and additionally:
		\begin{enumerate}
			\item[\textup{(v)}]
            For every accessible boundary \(b\in J\setminus J^\circ\) it holds:
				\[
                \text{If }b\text{ is absorbing, then }
				\int_{I} \, |x - \s (b)| \varphi^2 (x) \vd x < \infty
				\]
				for some (equivalently, for every) open interval \(I \subsetneq \s (J^\circ)\) with \(\s (b)\) as endpoint.
		\end{enumerate}
	\end{theorem}

\begin{proof}
As NUPBR implies NSA,
the conditions (i)--(iii) from Theorem~\ref{theo: main_NIP}
and condition~(iv) from Theorem~\ref{theo: main_NSA}
need to hold.
By Theorem~\ref{theo: FTAP_NUPBR},
it remains to prove that, under these conditions,
the above condition~(v) is equivalent to~\eqref{eq:310125a2}.
That is, from now on we assume (i)--(iii) from Theorem~\ref{theo: main_NIP}
and (iv) from Theorem~\ref{theo: main_NSA}.
To ease our presentation, we detail the proof only for the case $J^\circ=(\alpha,\infty)$, 
the general case being similar.

Consider the version of the IMPR $\theta$ that is given by~\eqref{eq:200125a1}.
We recall from~\eqref{eq:200125a3} that,
$\P$-a.s. for all $t\in[0,T_\alpha(Y)\wedge T)$ or even for all $t\in[0,T]$,
\begin{equation}\label{eq:310125a3}
\int_0^t \theta_s^2\vd\langle M,M\rangle_s
=
\int_{\s(J^\circ)}\varphi^2(x)L_t^x(\U)\vd x,
\end{equation}
depending on whether $\alpha$ is absorbing or whether $\alpha$ is inaccessible or reflecting, respectively.
Condition~(iv) from Theorem~\ref{theo: main_NSA} together with~\eqref{eq:310125a3}
immediately imply that,
if $\alpha$ is inaccessible or reflecting, then \eqref{eq:310125a2} holds
(notice that $\P$-a.s. the function $\s(J)\ni x\mapsto L_t^x(U)$ is bounded
as it is c\`adl\`ag with compact support in $\s(J)$).
It only remains to consider the case where $\alpha$ is an absorbing boundary point for~$Y$.

Thus, we now assume that $\alpha$ is absorbing.
First, observe that the formula~\eqref{eq:310125a3} holds on the whole interval $[0,T]$ also in this case.
Indeed, this follows from that both sides in~\eqref{eq:310125a3} are left-continuous (by monotone convergence),
stopped at $T_\alpha(Y)$ and coincide on $[0,T_\alpha(Y)\wedge T)$.
In particular, \eqref{eq:310125a2} is equivalent to
$$
\int_0^{T_\alpha(Y)\wedge T} \theta_s^2\vd\langle M,M\rangle_s < \infty \quad \P\text{-a.s.}, 
$$
which, by~\eqref{eq:310125a3}, is equivalent to 
\begin{equation}\label{eq:310125a4}
\int_{\s(\alpha)}^\infty\varphi^2(x)L_{T_{\s(\alpha)}(U)\wedge T}^x(\U)\vd x
<\infty \quad \P\text{-a.s.}, 
\end{equation}
where we used that $T_\alpha(Y)=T_{\s(\alpha)}(U)$.
We remark that, as $\alpha$ is accessible, we have $|\s(\alpha)|<\infty$.
Further notice that, as $\varphi \in L^2_\textup{loc} (\s (J^\circ))$,
the inequality~\eqref{eq:310125a4} holds $\P$-a.s. on $\{T_{\s(\alpha)}(U)>T\}$ because, on this event, the function $x\mapsto L_{T_{\s(\alpha)}(U)\wedge T}^x(U)$ has a compact support in $\s(J^\circ)$.
Summarizing, we established that \eqref{eq:310125a2} is equivalent to
\begin{equation}\label{eq:040225a1}
\int_{\s(\alpha)}^\infty\varphi^2(x)L_{T_{\s(\alpha)}(U)\wedge T}^x(\U)\vd x
<\infty\quad\P\text{-a.s. on }
\{T_{\s(\alpha)}(U)\le T\},
\end{equation}
and the remaining task is to characterize when \eqref{eq:040225a1} holds.
To this end, we realize the general diffusion $U$ on natural scale as a time-changed Brownian motion.
More precisely, by \cite[Theorem~33.9]{Kal21},
\begin{equation}\label{eq:040225a2}
\on{Law}\big(\U_t;\,t\in[0,T]\big)
=
\on{Law}\big(W_{\gamma(t)};\,t\in[0,T]\big),
\end{equation}
where \(W = (W_t)_{t \geq 0}\) is a one-dimensional Brownian motion starting in $\s(x_0)$
(on some probability space whose probability measure we denote by $\overline\P$)
and
\begin{align*}
\gamma(t)
&\triangleq
\inf \{ s \geq 0 \colon A_s > t \},
\\[1mm]
A_t
&\triangleq
\begin{cases}
\int_{(\s(\alpha),\infty)} L^x_t (W) \,\m^U (\rd x),
&t< T_{\s(\alpha)}(W),
\\
\infty,
&t\ge T_{\s(\alpha)}(W).
\end{cases}
\end{align*}
Notice that, in the formula for $A$, we explicitly encoded that $\alpha$ is absorbing for $Y$, i.e., $\s(\alpha)$ is absorbing for~$U$.
Due to~\eqref{eq:040225a2}, we obtain that \eqref{eq:040225a1} is equivalent to
\begin{equation}\label{eq:040225a3}
\int_{\s(\alpha)}^\infty\varphi^2(x)
L_{T_{\s(\alpha)}
(W_{\gamma(\cdot)})\wedge T}^x
(W_{\gamma(\cdot)})
\vd x
<\infty\quad\overline\P\text{-a.s. on }
\{T_{\s(\alpha)}
(W_{\gamma(\cdot)})
\le T\}.
\end{equation}
Observing that
$L_t^x(W_{\gamma(\cdot)})=L_{\gamma(t)}^x(W)$
and
$\gamma(
T_{\s(\alpha)}(W_{\gamma(\cdot)})
)
=T_{\s(\alpha)}(W)$,
we arrive at the equivalence between \eqref{eq:040225a3} and
\begin{equation}\label{eq:040225a4}
\int_{\s(\alpha)}^\infty\varphi^2(x)
L_{T_{\s(\alpha)}(W)}^x(W)
\vd x
<\infty\quad\overline\P\text{-a.s. on }
\{T_{\s(\alpha)}
(W_{\gamma(\cdot)})
\le T\}.
\end{equation}
By the semimartingale occupation time formula,
$$
\int_0^{T_{\s(\alpha)}(W)}
\varphi^2(W_t)\vd t
=
\int_{\s(\alpha)}^\infty
\varphi^2(x)
L_{T_{\s(\alpha)}(W)}^x(W)\vd x
\quad\overline\P\text{-a.s.},
$$
which means that \eqref{eq:040225a4} is equivalent to
\begin{equation}\label{eq:040225a5}
\int_0^{T_{\s(\alpha)}(W)}
\varphi^2(W_t)\vd t
<\infty\quad\overline\P\text{-a.s. on }
\{T_{\s(\alpha)}
(W_{\gamma(\cdot)})
\le T\}.
\end{equation}
In order to investigate when \eqref{eq:040225a5} holds, we recall from
\cite[Lemma 4.1]{MU_12_ECP} that
\begin{itemize}
\item[(a)]
if
$$
\int_I (x-\s(\alpha))\varphi^2(x)\vd x<\infty
$$
for some (equivalently, for every) open interval $I\subsetneq\s(J^\circ)$ with $\s(\alpha)$ as endpoint, then
$$
\int_0^{T_{\s(\alpha)}(W)}
\varphi^2(W_t)\vd t
<\infty\quad\overline\P\text{-a.s.};
$$

\item[(b)]
if
$$
\int_I (x-\s(\alpha))\varphi^2(x)\vd x=\infty
$$
for some (equivalently, for every) open interval $I\subsetneq\s(J^\circ)$ with $\s(\alpha)$ as endpoint, then
$$
\int_0^{T_{\s(\alpha)}(W)}
\varphi^2(W_t)\vd t
=\infty\quad\overline\P\text{-a.s.}
$$
\end{itemize}
Together with the fact that, by~\eqref{eq:040225a2} and \cite[Theorem 1.1]{BrugRuf16},
$$
\overline\P(
T_{\s(\alpha)}
(W_{\gamma(\cdot)})
\le T
)
=
\P(
T_{\s(\alpha)}(U)\le T
)>0,
$$
we can, finally, conclude that \eqref{eq:040225a5} is equivalent to condition~(v) from the formulation of the theorem.
This completes the proof.
\end{proof}

\begin{remark} \label{rem: NSA NUPBR (2)}
    Theorems \ref{theo: main_NSA} and~\ref{theo: main_NUPBR} recover the observation from Remark~\ref{rem: NSA NUPBR proof}. Indeed, if \(\Y\) has no absorbing boundary points, (v) is an empty condition and NSA and NUBPR are equivalent.
    In fact, under (i)--(iv), (v) fails precisely in the situation where~\eqref{eq:200125a2} holds while
    \[
    \int_0^T \theta^2_s \vd \langle M, M\rangle_s = \infty
    \]
    with positive \(\P\)-probability. 
\end{remark}

As in the previous sections, we present a corollary for the zero interest rate regime.

\begin{corollary}\label{coro: NUPBR r = 0}
For $r=0$ the following are equivalent:
\begin{enumerate}
\item[\textup{(i)}]
NUPBR holds.

\item[\textup{(ii)}]
The scale function $\s$ is a $C^1$-function on $J^\circ$ with a strictly positive (locally) absolutely continuous derivative,
$\s''/\s'\in L^2_\textup{loc}(J^\circ)$, and
every accessible for $Y$ boundary point $b\in J\setminus J^\circ$ is absorbing
and such that
$$
\int_I | x - b | (\s''/\s')^2 (x) \vd x < \infty
$$
for some (equivalently, for every)
open interval \(I \subsetneq J^\circ\) with \(b\) as endpoint.

\item[\textup{(iii)}]
The scale function $\s$ admits the representation~\eqref{eq:220125a3}
with some $\mu\in L^2_\textup{loc}(J^\circ)$, and
every accessible for $Y$ boundary point $b\in J\setminus J^\circ$ is absorbing
and such that
\begin{equation}\label{eq:050225a1}
\int_I | x - b | \mu^2 (x) \vd x < \infty
\end{equation}
for some (equivalently, for every)
open interval \(I \subsetneq J^\circ\) with \(b\) as endpoint.

\item[\textup{(iv)}]
The inverse scale function $\q$ is a $C^1$-function on $\s(J^\circ)$ with a strictly positive (locally) absolutely continuous derivative,
$\q''/\q'\in L^2_\textup{loc}(\s(J^\circ))$, and
every accessible for $Y$ boundary point $b\in J\setminus J^\circ$ is absorbing
and such that
$$
\int_I | x - \s(b) | (\q''/\q')^2 (x) \vd x < \infty
$$
for some (equivalently, for every)
open interval \(I \subsetneq \s(J^\circ)\) with \(\s(b)\) as endpoint.

\item[\textup{(v)}]
The inverse scale function $\q$ admits the representation~\eqref{eq:220125a4}
with some $\nu\in L^2_\textup{loc}(\s(J^\circ))$, and
every accessible for $Y$ boundary point $b\in J\setminus J^\circ$ is absorbing
and such that
\begin{equation}\label{eq:050225a2}
\int_I | x - \s(b) | \nu^2 (x) \vd x < \infty
\end{equation}
for some (equivalently, for every)
open interval \(I \subsetneq \s(J^\circ)\) with \(\s(b)\) as endpoint.
\end{enumerate}
\end{corollary}

It is worth observing that Corollary~\ref{coro: NUPBR r = 0} recovers \cite[Theorem~3.10]{CU23}
(more precisely, the latter is the equivalence (i) $\Longleftrightarrow$ (iii)).
We highlight that the proofs in the present paper are quite different from those in \cite{CU23}.
In particular, we now have a second proof for \cite[Theorem~3.10]{CU23} that does not rely on the concept of separating times for general diffusions (\cite{CU22}) that was used in \cite{CU23} in a crucial manner.

\begin{proof}
The equivalence (i) $\Longleftrightarrow$ (v) is easily built on the corresponding equivalence in Corollary~\ref{coro: NSA r = 0}.
We only need to add the obvious argument relating claim~(v) in Theorem~\ref{theo: main_NUPBR} with~\eqref{eq:050225a2}.
The equivalences
(ii) $\Longleftrightarrow$ (iii)
and
(iv) $\Longleftrightarrow$ (v)
are straightforward.
It remains to prove the equivalence
(iii) $\Longleftrightarrow$ (v).
The equivalence between
(iii) without~\eqref{eq:050225a1}
and
(v) without~\eqref{eq:050225a2}
follows from Corollary~\ref{coro: NSA r = 0}.
In the following, we thus assume that
(iii) without~\eqref{eq:050225a1}
(equivalently, (v) without~\eqref{eq:050225a2})
is satisfied.
In particular, we can use the functions $\mu$ and $\nu$ from
\eqref{eq:220125a3} and~\eqref{eq:220125a4}.
Let $b\in J\setminus J^\circ$ be an accessible boundary point for $Y$.
The aim is to prove that
\eqref{eq:050225a1} is equivalent to~\eqref{eq:050225a2}.
To this end, we recall the formulas
\begin{align}
\q'(x)
&=
\frac1{\s'(\q(x))}\quad\text{for all }x\in\s(J^\circ),
\label{eq:050225a4}\\[1mm]
\nu(x)
&=
-\mu(\q(x))\q'(x)\quad\text{for }\llambda\text{-a.a. }x\in\s(J^\circ)
\label{eq:050225a5}
\end{align}
(cf.~\eqref{eq:050225a3}), take an open interval $I\subsetneq J^\circ$ with $b$ as endpoint and compute
\begin{align}
\int_{\s(I)}
|x-\s(b)|\nu^2(x)\vd x
&=
\int_{\s(I)}
|x-\s(b)|\mu^2(\q(x))(\q'(x))^2\vd x
\label{eq:050225a6}\\[1mm]
&=
\int_I
\frac{|\s(y)-\s(b)|\mu^2(y)}{\s'(y)}\vd y,
\notag
\end{align}
where we used \eqref{eq:050225a4}--\eqref{eq:050225a5} together with the substitution $x=\s(y)$, $y\in I$.
It remains to notice that,
under the condition
$\mu\in L^2_\textup{loc}(J^\circ)$ (which is assumed),
we have the equivalence
$$
\int_I
\frac{|\s(y)-\s(b)|\mu^2(y)}{\s'(y)}\vd y<\infty
\;\;\Longleftrightarrow\;\;
\int_I
|y-b|\mu^2(y)\vd y<\infty
$$
(see the equivalencies (24) $\Longleftrightarrow$ (25) and (26) $\Longleftrightarrow$ (27) from \cite{MU12_PTRF}).
Together with~\eqref{eq:050225a6},
this yields the required equivalence
\eqref{eq:050225a1} $\Longleftrightarrow$ \eqref{eq:050225a2}.
This concludes the proof.
\end{proof}

\section{Discussions and Examples} \label{sec_DE}
In this section, (i) we discuss the influence of the no-arbitrage notions NIP, NSA and NUPBR on the boundary behavior of the diffusion \(Y\), (ii) we highlight their influence on the regularity of the scale function and (iii) we take a particular look at interrelations. 
Our findings reveal several surprising observations that contrast with existing results in the literature.
Specifically, we show that, even in the case where $r=0 $, the presence of reflecting boundaries does not necessarily rule out NIP. Also, we show that skewness and stickiness can ``cancel'' each other, allowing even the strongest of the three no-arbitrage conditions, NUPBR, to hold. 
This latter effect is observed only in non-zero interest rate environments.

\subsection{Influence of Reflecting Boundaries}\label{sec: examples reflection}

In this section we answer the question to which degree the presence of reflecting boundaries entails certain arbitrage notions. 
As we will explain in the following, the presence of non-zero interest rates has a major influence on the answer to this question.

\subsubsection{The Zero Interest Rate Regime}
Let us start with a discussion of the zero interest rate regime, i.e., the case \(r = 0\). For this setting, it was proved in \cite{CU23} that NUPBR cannot hold in the presence of reflecting boundaries, see also Corollary~\ref{coro: NUPBR r = 0} for a different proof. Furthermore, Corollary~\ref{coro: NSA r = 0} shows that even the weaker NSA condition forces the absence of reflecting boundaries
(also see Discussion~\ref{disc:040225a2}~(v)).
Similar observations for different no-arbitrage concepts were also made in~\cite{BDH23,MelWan21}. In particular, the paper \cite{BDH23} shows that even the weaker NIP condition fails for a geometric Brownian motion model with reflection (at a strictly positive reflecting boundary).
All these results might lead to the impression that the absence of reflecting boundaries is a necessary structural requirement
even for very weak forms of absence of arbitrage.
The following example makes the surprising observation that reflecting boundaries do not exclude the weak notion NIP, which adds a new layer to the picture drawn in \cite{BDH23}.

\begin{example}
\label{ex: bessel}
		Take \(r = 0\) and suppose that \(\Y\) is a square Bessel process of dimension \(\delta \in (0, 2)\) with starting value \(x_0 > 0\), i.e., a solution process to the stochastic differential equation 
		\begin{align*}
			\rd \Y_t = 2 \sqrt{|\Y_t|} \vd W_t + \delta \vd t, \quad \Y_0 = x_0, 
		\end{align*}
		where \(W\) is a standard Brownian motion. 
		It is well-known that \(\Y\) is a (general) diffusion with state space \(\bR_+\) for which the origin is instantaneously reflecting (cf. \cite[Proposition~XI.1.5]{RevYor}).
        In particular,
		\(
		\1_{\{\Y_t = 0\}} \vd t= 0
		\)
		(see also \cite[V.48.6 (ii)]{RogWilV2})
		and consequently,
		\[
		\rd \Y_t = 2 \sqrt{\Y_t} \vd W_t + \delta \1_{\{\Y_t \not = 0\}} \vd t, \quad \Y_0 = x_0.
		\]
		Now, it is easy to see that 
		\[
		\theta \triangleq \frac{\delta \1_{\{\Y \not = 0\}}}{4 \Y}
		\]
		is an IMPR and the NIP condition holds by Theorem~\ref{theo: FTAP_NIP}.
		
		This observation is in line with Theorem~\ref{theo: main_NIP}. Indeed,
the scale function \(\s\) associated to \(\Y\) is given by 
$$
\s (x) = x^{1 - \delta / 2}, \quad x\in\bR_+,
$$
		cf. \cite[V.48.6 (iii)]{RogWilV2}.
		Clearly, we have
		\begin{align*}
			\q (x) = x^{1 / (1 - \delta / 2)}, \quad x \in \s (\bR_+) = \bR_+.
		\end{align*} 
As $1/(1-\delta/2)>1$, it follows that \(\q' (\s (0)) = \q' (0) = 0\), i.e.,
(i) of Corollary~\ref{coro: NIP r = 0} is satisfied.
Obviously, (ii) of Corollary~\ref{coro: NIP r = 0} holds as well and, as a consequence, NIP holds.
	\end{example}

    \subsubsection{The Non-Zero Interest Rate Regime}
    Next, we discuss what changes in the presence of non-zero interest rates, i.e., in the case \(r \ne 0\). Surprisingly, it turns out that in this case even NUPBR might hold in the presence of reflecting boundaries. This situation is illustrated by the following example. 

    \begin{example} \label{ex: NUBPR ref interest rate}
    Assume that \(\Y\) is a Brownian motion with (instantaneous or sticky) reflection at one and a starting value \(x_0 > 1\).
More precisely, we consider the case \(J = [1, \infty)\) and
			\[
			\s (x) = x, \quad \m (\rd x) = \rd x + \rho \delta_{1} (\rd x)
			\]
with a parameter \(\rho \ge 0\).
Notice that, if $\rho=0$, then the reflection in $1$ is instantaneous, while in the case $\rho>0$, the boundary $1$ is sticky reflecting.
			We observe that 
			\begin{align*}
				2r\rho = 1 \quad \Longleftrightarrow \quad \textup{NIP} \quad \Longleftrightarrow \quad \textup{NSA} \quad \Longleftrightarrow \quad \textup{NUPBR}.
			\end{align*}
			Indeed, (i.b) from Theorem~\ref{theo: main_NIP} is precisely the equality \(2r \rho = 1\), and (ii) and (iii) hold trivially. This shows the first equivalence.
            The function $\varphi$ from Theorem~\ref{theo: main_NSA} is given by the formula
            \(\varphi (x) = - r x\), \(x \geq 1\).
            Therefore, the additional (local) integrability requirements from the Theorems~\ref{theo: main_NSA} and~\ref{theo: main_NUPBR} clearly hold as well.
            This proves that \(2 r \rho = 1\) implies even the strongest notion NUPBR, which completes the picture.
    \end{example}
    
Example~\ref{ex: NUBPR ref interest rate} is different from Example~\ref{ex: bessel}, as it requires \(r \neq 0\) for NIP, NSA and NUPBR to hold (with an appropriately chosen value for the stickiness parameter $\rho$).
For the case $r=0$, Example~\ref{ex: NUBPR ref interest rate} also shows that NIP fails for the Bachelier model with reflection (regardless of whether instantaneous or sticky).
For this model, the non-existence of an equivalent martingale measure was explained in~\cite{MelWan21}.
Our discussion sharpens this observation in the sense that even NIP fails.

\subsubsection{Summary}
The following table gives an overview on the observations discussed in Section~\ref{sec: examples reflection}.
\begin{center}
\begin{tabular}{c| c | c | c}
& NIP & NSA & NUPBR
\\ 
\hline 
\(r = 0\) & possible to have & no reflecting & no reflecting
\\
& reflecting boundaries & boundaries possible & boundaries possible
\\
\hline 
\(r \not = 0\) & possible to have & possible to have & possible to have 
\\ 
& reflecting boundaries & reflecting boundaries & reflecting boundaries
\end{tabular} 
\end{center}

\subsection{Implications for the Regularity of the Scale Function}\label{sec: examples regularity}

In this section we want to discuss the question how much regularity of the scale function is entailed by the notions NIP, NSA and NUPBR. As in Section~\ref{sec: examples reflection}, we distinguish between the zero and non-zero interest rate regimes.

\subsubsection{The Zero Interest Rate Regime}
Let us start by a discussion of the zero interest rate regime, i.e., the case \(r = 0\).
One interesting observation from the paper \cite{CU23} is that NUPBR forces the scale function \(\s\) to be continuously differentiable with locally absolutely continuous derivative, which means that it is of the same form as the scale function of a diffusion that is given through a stochastic differential equation. An improvement of this observation is provided by Corollary~\ref{coro: NSA r = 0}, which shows that even the weaker NSA condition entails such a structure of the scale function. However, as the following example illustrates, this is not the case for the weaker NIP condition.

\begin{example}\label{ex: x^3}
Let \(W\) be a one-dimensional Brownian motion starting at \(x_0 \not = 0\) and consider the diffusion semimartingale \(Y \triangleq W^3\). In this situation, the scale function is given by 
\[
\s (x) = \begin{cases}
    - (- x)^{1/3}, & x \leq 0, \\ x^{1/3}, & x \geq 0, 
\end{cases}
\]
which is not continuously differentiable on $\bR$
(even not the difference of two convex functions on~$\bR$),
as the right (and also the left) derivative in zero explodes. 

In this case, we have $\q(x)=x^3$ for $x\in\bR$ and 
Corollary~\ref{coro: NIP r = 0} immediately implies that NIP holds,
while Corollary~\ref{coro: NSA r = 0} yields that NSA, hence also NUPBR, are violated
(the function $\q'$ is \emph{not} strictly positive,
cf. Corollary~\ref{coro: NSA r = 0} (iv) or~(v)).

It is also instructive to show this via methods from stochastic calculus.
It\^o's formula yields that 
\[
\rd Y_t = 3 W^2_t \vd W_t + 3 W_t \vd t.
\]
Now, it is easy to see that \(\theta \triangleq \1_{\{W \ne 0\}} / (3 W^3)\) is an IMPR, which implies that NIP holds by Theorem~\ref{theo: FTAP_NIP}.
In this example, we have
$\theta_s^2\vd\langle M,M\rangle_s=(\1_{\{W_s\ne0\}}/W_s^2)\vd s$.
By \cite[Theorem 2.6]{MU_12_ECP} and strong Markov property of Brownian motion,
\[
\P \Big( \int_{T_0 (W)}^{T_0 (W) + \varepsilon} \frac{\1_{\{W_s\ne0\}}}{W^2_s} \vd s = \infty, \, \forall \, \varepsilon > 0 \Big) = 1.
\]
As a consequence, \(\P\)-a.s. on \(\{T_0 (W) < T\}\)
\[
\inf \Big\{ t \in [0, T) \colon \int_t^{t + h} \theta^2_s \vd \langle M, M\rangle_s = \infty, \, \forall \, h \in (0, T - t] \Big\} \leq T_0 (W) < T.
\]
As \(\P (T_0 (W) < T) > 0\),
NSA (and, consequently, also NUPBR) fails by Theorem~\ref{theo: FTAP_NSA}.

We, finally, discuss the point mentioned in the end of Remark~\ref{rem:090325a1}~(ii),
i.e.,
we now prove that \eqref{eq:090325a1} holds (notice that we already observed that \eqref{eq:200125a2} is violated).
In this example, the mean-variance tradeoff process $K=(K_t)_{t\in[0,T]}$ is given by the formula
$$
K_t=\int_0^\cdot \frac{\1_{\{W_s\ne0\}}}{W_s^2}\vd s.
$$
By \cite[Lemma 4.1]{MU_12_ECP}, it holds
$$
\lim_{t\nearrow T_0(W)}K_t=\infty\quad\P\text{-a.s.},
$$
hence \eqref{eq:090325a1} is satisfied (that is, the mean-variance tradeoff process does not jump to infinity
but reaches infinity in a continuous way).
\end{example}

It is also interesting to compare
the Corollaries \ref{coro: NIP r = 0}, \ref{coro: NSA r = 0} and~\ref{coro: NUPBR r = 0}
from the viewpoint of regularity requirements on the inverse scale function $\q$.
In the following table we discuss only the regularity implied by the no-arbitrage notions
and do not mention requirements of other types
(integrability and boundary behavior):
\begin{center}
\begin{tabular}{c| c | c}
& NIP & NSA, NUPBR
\\ 
\hline 
regularity of $\q$ & $\q$ is $C^1$ on $\s(J^\circ)$ with a & $\q$ is $C^1$ on $\s(J^\circ)$ with a
\\
in the case $r=0$ & locally absolutely continuous $\q'$ & \emph{strictly positive}
\\
& & locally absolutely continuous $\q'$
\end{tabular} 
\end{center}
Given that $\q$ is always strictly increasing,
the difference between the second and the third columns
(namely, the words ``strictly positive'')
looks somewhat empty at first glance.
But Example~\ref{ex: x^3} shows that this difference really matters.

In Example~\ref{ex: x^3}, the set $\{\q'=0\}$ consists only of one point, which is already enough to violate NSA.
On the other hand, in the following example,
the set $\{\q'=0\}$ has a positive Lebesgue measure,
while NIP is still satisfied.
As mentioned in Discussion~\ref{disc:040225a1}~(iii), this also means that the RP fails for $Y$ from the following example.

\begin{example}\label{ex:060225a1}
The idea is taken from \cite[Lemma~2.1]{CU24}.
To make the example self-contained, we recall some details.

First, we take a closed set $F\subset[0,1]$ with empty interior such that $\llambda(F)>0$.
This could be a \emph{fat} Cantor set or, alternatively, one could construct such a set as follows (cf. \cite[Example~1.7.6]{bogachev}).
Let $\{q_n\colon n\in\mathbb N\}$ be an enumeration of all rational points in $[0,1]$.
Take $a\in(0,1)$ and a sequence $\{r_n\colon n\in\mathbb N\}$ such that
$\sum_{n=1}^\infty 2r_n\le a$.
It is easy to verify that
$F\triangleq[0,1]\setminus G$,
where $G\triangleq\bigcup_{n\in\mathbb N}(q_n-r_n,q_n+r_n)$,
satisfies the requirements.

Next, consider a one-dimensional Brownian motion $W$ starting at \(x_0\in\bR\) and set \(Y \triangleq \q(W)\), where
$$
\q (x) \triangleq \int_0^x d_F (z) \vd z,
\quad x \in \bR,
\quad d_F (z) \triangleq \inf_{y \in F} |z - y|.
$$
Notice that $\q$ is a $C^1$-function on $\bR$ with
\begin{equation}\label{eq:070225a1}
\{x\in\bR\colon\q'(x)=0\}=F
\end{equation}
(because $z\mapsto d_F (z)$ is continuous and $F$ is closed) and
$\q$ is strictly increasing (because $F$ is closed and does not contain any open interval).
Then $Y$ is a general diffusion with state space
$\q(\bR)=\bR$, scale function $\s\triangleq\q^{-1}$ and speed measure $\llambda\circ\q^{-1}$.
Furthermore, $\q'=d_F(\cdot)$ is Lipschitz continuous on $\bR$, hence absolutely continuous.
In particular, $\q'$ has a locally finite variation, showing that $\q$ is the difference of two convex functions on $\bR$.
Consequently, Standing Assumption~\ref{SA: semi + boundary} is satisfied.

Here, like in the previous example,
Corollary~\ref{coro: NIP r = 0} immediately implies that NIP holds,
while Corollary~\ref{coro: NSA r = 0} yields that NSA, hence also NUPBR, are violated.
In contrast to the previous example, by~\eqref{eq:070225a1}, we now have
$\llambda(\{\q'=0\})>0$.

This example can also be continued in the direction of even more regularity on $\q$, but still failure of NSA and NUPBR:
Namely, with a similar idea one can even construct a strictly increasing \(C^\infty\)-function $\q$ whose derivatives have a zero set of positive Lebesgue measure, cf. \cite[Remark~2.3]{CU24}.
\end{example}

\subsubsection{The Non-Zero Interest Rate Regime}

Next, we consider the non-zero interest rate regime, which turns out to be fundamentally different from its zero interest counterpart. Indeed, as the following example shows, even the strongest NUPBR condition does not force the scale function to be continuously differentiable. 

	\begin{example} \label{ex: sticky skew}
		Consider a sticky-skew Brownian motion model $Y$ that has state space \(J = \bR\), scale function \(\s\) and speed measure \(\m\) defined by
		$$
				\s(x) = (x-\xi) v_\kappa (x)
                \quad\text{and}\quad
				\m (\rd x) = \frac{\rd x}{v_\kappa (x)} + c\, \delta_{\xi} (\rd x),
		$$
		where 
		\(\xi \in \bR\) is the sticky-skew point, \(\kappa \in (0, 1)\) is the skewness parameter, \(c > 0\) is the stickiness parameter and 
        		\[
		v_\kappa (x) = \begin{cases} 1 - \kappa, & x > \xi, \\ \kappa, & x \leq \xi. \end{cases}
		\]
		We easily compute that
		\[
		\q (x) = \begin{cases} \xi + \frac{x}{1 - \kappa}, & x \geq 0, 
        \\ \xi + \frac{x}{\kappa}, & x < 0,\end{cases}
		\]
		and 
		\[
		\q_+' (x) = \begin{cases} \frac{1}{1 - \kappa}, & x \geq 0, 
        \\ \frac{1}{\kappa}, & x < 0.\end{cases} 
		\]
		Evidently, we also get that 
		\begin{align*}
			\q'' (\rd x) =  \frac{2 \kappa - 1}{(1 - \kappa) \kappa}\, \delta_0 (\rd x).
		\end{align*}
        Straightforward calculations yield
        $$
        \m^U(\rd x)=\frac{\rd x}{a_\kappa(x)}+c\,\delta_0(\rd x),
        $$
        where
        $$
        a_\kappa(x)=\begin{cases} (1 - \kappa)^2, & x > 0, \\ \kappa^2, & x \leq 0, \end{cases}
        $$
which implies
		\[
		\m_{\ac}^U (x) = \frac{1}{a_\kappa (x)}, \quad \m_{\on{si}}^U (\rd x) = c \, \delta_0 (\rd x).
		\] 
		It follows
        that (i) and~(iii) from Theorem~\ref{theo: main_NIP} always hold
        (part~(iii) holds because \(\{\q'_+ = 0\} = \emptyset\)) and (ii) holds if and only if 
		\begin{align} \label{eq: sticky skew condition}
			r \xi c = \frac{2 \kappa - 1}{2 \kappa (1 - \kappa)}.
		\end{align} 
		In summary, we conclude that NIP holds if and only if \eqref{eq: sticky skew condition} is satisfied. 
Moreover, as \(\varphi\) from Theorem~\ref{theo: main_NSA} is given by
\(
\varphi = - r \q \m^U_\ac/ \q'_+ \in L^2_\textup{loc} (\bR), 
\)
using also Theorems \ref{theo: main_NSA} and~\ref{theo: main_NUPBR}, we get that NUPBR \(\Longleftrightarrow\) NSA \(\Longleftrightarrow\) NIP \(\Longleftrightarrow\) \eqref{eq: sticky skew condition}.

		\smallskip 
		Condition~\eqref{eq: sticky skew condition} can also be established by stochastic calculus methods. To explain this,  in the following lemma, we recall that a sticky-skew Brownian motion can be realized via a stochastic equation with local time constraint. The existence and uniqueness parts are known (cf., e.g., \cite{WT21}). We provide the computations of scale and speed in Appendix~\ref{app: pf lemma}.
        
		\begin{lemma} \label{lem: sticky-skew BM}
		For any $x_0\in\bR$,
            the system 
			\begin{align*}
				\rd \widehat{\Y}_t &= \1_{\{\widehat{\Y}_t \not = \xi\}} \vd W_t + \frac{2 \kappa - 1}{2 \kappa} \, \rd L^\xi_t (\widehat{\Y}), \\
				\1_{\{\widehat{\Y}_t = \xi \}} \vd t&= c\, (1 - \kappa) \, \rd L^\xi_t (\widehat{\Y}) 
			\end{align*}
            with initial condition $\widehat Y_0=x_0$
			satisfies weak existence and uniqueness in law. Furthermore, the unique in law solution $\widehat Y$ is a general diffusion with scale function \(\s\) and speed measure \(\m\) as defined above.
		\end{lemma}
        
		Given the dynamics of \(\widehat{\Y}\) from Lemma~\ref{lem: sticky-skew BM}, we observe that 
		\[
		\rd ( e^{-rt} \widehat{\Y}_t ) = e^{-rt} \Big( \1_{\{\widehat{\Y}_t \not = \xi\}} \vd W_t - r \widehat{\Y}_t \1_{\{\widehat{\Y}_t \not = \xi\}} \vd t+ \Big( \frac{2 \kappa - 1}{2 \kappa c (1 - \kappa)}  - r \xi \Big) \, \1_{\{\widehat{Y}_t = \xi \}} \vd t\Big).
		\] 
		It is not hard to conclude from this formula that, for \(S_t \equiv e^{-rt}\widehat{Y}_t\), the characterization of NIP from Theorem~\ref{theo: FTAP_NIP} holds if and only if the \(\1_{\{\widehat{Y}_t = \xi\}} \vd t\) part vanishes, which is precisely the case when \eqref{eq: sticky skew condition} is satisfied. 
	\end{example}

\subsubsection{Summary}
The following table summarizes the most important observations from Section~\ref{sec: examples regularity}.
\begin{center}
\begin{tabular}{c| c | c | c}
& NIP & NSA & NUPBR
\\ 
\hline 
\(r = 0\) & scale function possibly & scale function & scale function
\\
& less regular than \(C^1\) & at least \(C^1\) & at least \(C^1\)
\\
\hline 
\(r \not = 0\) & scale function possibly & scale function possibly & scale function possibly
\\ 
& less regular than \(C^1\) & less regular than \(C^1\) & less regular than \(C^1\)
\end{tabular} 
\end{center}

\subsection{Discussion of Relations between NIP, NSA and NUPBR} \label{sec: discussion}

As mentioned at the end of Section~\ref{sec: NIP, NSA, NUPBR intro}, the following general implications hold:
    \begin{align} \label{eq: connection notions}
	\textup{NUPBR} \implies \textup{NSA} \implies \textup{NIP}.
\end{align}
Two following corollaries of our main results refine this picture in our general diffusion framework. 

\begin{corollary}\label{cor:070225a1}
For our general diffusion setting, all implications in \eqref{eq: connection notions} are strict in the sense that, for each implication from \eqref{eq: connection notions} and any \(r \in \bR\), there exists a general diffusion \(Y\) such that the
converse
implication fails. 
\end{corollary}

\begin{proof}
We provide counterexamples. Consider a ``generalized square Bessel process'' \(Y\) of dimension \(\delta \in (0, 2)\), i.e., a general diffusion with state space \(J = \bR_+\), scale function 
\[
\s (x) = x^{- \nu}, \quad \nu \triangleq \frac{\delta}{2} - 1 \in (-1, 0), \quad x \in \bR_+, 
\]
and speed measure
\[
\m (\rd x) = \frac{x^\nu}{4 | \nu |}\vd x \text{ on } \mathcal{B} ((0, \infty)), \quad \m (\{0\}) \in [0, \infty].
\] 
We use the branding ``generalized square Bessel process'', because we also allow the origin to be sticky reflecting or absorbing, while it is instantaneously reflecting in the classical case, cf. \cite[Section~XI.1]{RevYor}.
Take an arbitrary $r\in\bR$.
It is routine to check that (i)--(iii) from Theorem~\ref{theo: main_NIP} are satisfied. Thus, NIP holds.
Furthermore, straightforward computations reveal that 
\begin{align*}
\varphi (x) 
&= \Big[ \, \frac{1}{2} \Big( \frac{1}{|\nu|} - 1 \Big) - \frac{r}{4 |\nu|}x^{\frac{1}{|\nu|}} \, \Big] \, \frac{1}{x} \approx \text{const}\, \frac{1}{x}, \quad x \searrow 0. 
\end{align*} 
Hence, by Theorem~\ref{theo: main_NSA}, NSA holds if and only if the origin is absorbing, i.e., \(\m (\{0\}) = \infty\). This shows that the implication NSA \(\implies\) NIP is strict.
Further, in the absorbing case \(\m (\{0\}) = \infty\), Theorem~\ref{theo: main_NUPBR} shows that NUPBR fails. Consequently, the implication NUPBR \(\implies\) NSA is also strict. 
\end{proof}

The following result sharpens the observations from Remarks \ref{rem: NSA NUPBR proof} and~\ref{rem: NSA NUPBR (2)}.

\begin{corollary}\label{cor:070225a2}
Suppose that each boundary of \(Y\) is either inaccessible or reflecting.
Then, for any $r\in\bR$, NUBPR \(\Longleftrightarrow\) NSA. The implication NSA \(\implies\) NIP remains strict
in the sense that, for any $r\in\bR$, there exists a general diffusion $Y$ such that the
converse
implication fails.
\end{corollary}

\begin{proof}
If all boundaries of \(Y\) are inaccessible or reflecting, then Theorems \ref{theo: main_NSA}
and~\ref{theo: main_NUPBR} entail the claimed equivalence NUBPR \(\Longleftrightarrow\) NSA.
The example from the proof of Corollary~\ref{cor:070225a1} with $\m(\{0\})<\infty$
(that is, $0$ is reflecting)
shows that the implication NSA \(\implies\) NIP remains strict.
Alternatively, this follows from Example~\ref{ex: x^3}
(notice that both boundaries are inaccessible there),
which extends readily to the case with non-zero interest rate. 
\end{proof}

	\appendix	
	\section{Compensation with Reflecting Boundaries} 
	\label{appendix_A: proof_lem: comp}
The following lemma is a restatement of \cite[V.47.23 (ii)]{RogWilV2}. We provide here a different proof based on a change of time argument that we think is quite instructive.

	\begin{lemma}
   Take $\l\in\bR$.
    Let \(U\) be a general diffusion (for a given right-continuous filtration~\(\mathbf{F}\)) on natural scale
    starting at some $x_0\ge\l$
    with state space \([\l, \infty)\) and reflecting boundary \(\l\). Then, the process 
		\[
		U - \frac{L^{\l} (U)}{2}
		\]
		is a local martingale. 
	\end{lemma} 
	
	\begin{proof}
		Let \(W = (W_t)_{t \geq 0}\) be a standard Brownian motion starting at \(x_0\) and set \(Z \triangleq \l + |W - \l|\).
       By $\m^U(\rd x)$ we denote the speed measure of $U$ (on $\mathcal B([\l,\infty))$)
        and by $\bfFW$ (resp., $\bfFZ$) the right-continuous filtration generated by $W$ (resp., by~$Z$). Notice that $\bfFZ\subset\bfFW$.
		By \cite[Theorem~16.56]{breiman}, the process 
		\[
		Z_L = (Z_{L_t})_{t \geq 0}, \quad L_t \triangleq \inf \Big\{ s \geq 0 \colon \int L^x_s (Z) \, \m^{\U} (\rd x) > t \Big\}, 
		\]
		has the same law as \(U\), when considered on the infinite time horizon.
        We observe that \(t \mapsto L_t\) is a time change w.r.t. $\bfFZ$, and we emphasize that it is a.s. continuous,
        see the discussion preceding \cite[Theorem~16.56]{breiman}.
        Tanaka's formula (\cite[Theorem~IV.43.3]{RogWilV2}) yields that 
		\begin{align*}
			d Z_{t} = \on{sgn} (W_t - \l) \vd W_t + \rd L^{\l}_t (W).
		\end{align*}
By \cite[Exercise~VI.1.17]{RevYor} together with continuity of the Brownian local time in the space variable, we have
$L_t^\l(W)=L_t^\l(Z)/2$.
Therefore,
$$
M\triangleq Z-\frac{L^\l(Z)}2
$$
is a continuous local martingale w.r.t. $\bfFW$,
hence also w.r.t. $\bfFZ$.
By \cite[Exercise~VI.1.27]{RevYor}, we have
$L^\l_{L_t}(Z)=L^\l_t(Z_L)$.
Now, \cite[Proposition~V.1.5]{RevYor} yields that the process
$$
M_L\equiv Z_L-\frac{L^\l(Z_L)}2
$$
is a continuous local martingale w.r.t. the time-changed filtration $\bfFZ_L$, hence also w.r.t. the (smaller) right-continuous filtration generated by the time-changed process $Z_L$,
        as it is adapted to the latter filtration and has continuous paths. Now, by virtue of \cite[Theorem~10.37]{jacod79}, as \(Z_L\) has the same law as \(U\), we can conclude that \(U - L^{\l} (U) / 2\) is a local martingale for its natural filtration. 
		Finally, \cite[Lemma~5.13]{CU23} yields that this local martingale property also transfers to the
        larger
        filtration~\(\mathbf{F}\).
        For this step, it is important that \(U\) is a Markov process w.r.t.~$\mathbf F$.
        The proof is complete.
	\end{proof}

	\section{Proof of Lemma~\protect\ref{lem: sticky-skew BM}}\label{app: pf lemma}
    
	For reader's convenience, we recall the statement of Lemma~\ref{lem: sticky-skew BM}.
	
	\begin{lemmaOhne}
    Let $x_0$ and $\xi$ be real numbers.
		The system 
		\begin{align}
			\rd \widehat{\Y}_t &= \1_{\{\widehat{\Y}_t \not = \xi\}} \rd W_t + \frac{2 \kappa - 1}{2 \kappa} \, \rd L^\xi_t (\widehat{\Y}),
            \label{eq:120325a1}\\
			\1_{\{\widehat{\Y}_t = \xi \}} \vd t&= c\, (1 - \kappa) \, \rd L^\xi_t (\widehat{\Y})
            \label{eq:120325a2}
		\end{align}
        with initial condition $\widehat Y_0=x_0$
		satisfies weak existence and uniqueness in law. Furthermore, the unique in law solution $\widehat Y$ is a general diffusion with scale function \(\s\) and speed measure \(\m\) that are given by 
		\[
		\s (x) = (x-\xi)v_\kappa(x)
        \quad\text{and}\quad
        \m (\rd x) = \frac{\rd x}{v_\kappa (x)} + c\, \delta_{\xi} (\rd x),
		\]
		where
		\[
		v_\kappa (x) = \begin{cases} 1 - \kappa, & x > \xi, \\ \kappa, & x \leq \xi. \end{cases} 
		\] 
	\end{lemmaOhne} 
	
	\begin{proof}
		The existence and uniqueness parts are known and can for example be deduced from \cite[Theorem~2.1]{WT21}. From the construction (which is based on a homeomorphic change of space and an It\^o--McKean type change of time), it also follows that \(\widehat{Y}\) is a general diffusion. 
		In the following, we provide the computations for scale and speed.  
		Let us start with the scale function. 
		An application of the generalized It\^o formula (\cite[Theorem~IV.45.1]{RogWilV2}), and using the fact that \(\rd L^\xi_t (\widehat{Y})\) is supported on \(\{t \colon \widehat{Y}_t = \xi\}\) (\cite[Proposition~VI.1.3]{RevYor}), yields that 
		\begin{align*}
			\rd\s (\widehat{Y}_t) &= v_\kappa (\widehat{Y}_t) \, \rd \widehat{Y}_t + \frac{1 - 2 \kappa}{2}\, \rd L^\xi_t (\widehat{Y}) 
			\\&= v_\kappa (\widehat{Y}_t) \, \1_{\{\widehat{Y}_t \not = \xi \}} \, \rd W_t + v_\kappa (\xi) \, \frac{2 \kappa - 1}{2 \kappa} \vd L^\xi_t (\widehat{Y}) + \frac{1 - 2 \kappa}{2}\, \vd L^\xi_t (\widehat{Y}) 
			\\&= v_\kappa (\widehat{Y}_t) \, \1_{\{\widehat{Y}_t \not = \xi \}} \, \rd W_t. \phantom \int 
		\end{align*} 
		As a consequence, by \cite[Proposition~VII.3.5]{RevYor}, \(\s\) is a scale function for \(\widehat{Y}\). 
		
		Next, we prove that \(\m\) is the corresponding speed measure. By virtue of \cite[Exercise~VII.3.18]{RevYor}, it suffices to prove that \(\widehat\m \triangleq \m \circ \s^{-1}\) is the speed measure of \(\s (\widehat{Y})\), which is a diffusion on natural scale. 
		We use a martingale problem argument. Take a function  \(f \in C_b (\bR; \bR)\) such that \(f'_+\) exists, is finite, right-continuous, locally of bounded variation and that \(\vd f'_+ = 2g\vd \widehat\m\) for some \(g \in C_b (\bR; \bR)\).
        Using the generalized It\^o formula
        in the first equality below,
        \cite[Exercise~VI.1.23]{RevYor}
        in the third equality,
        the occupation time formula
        in the fourth equality
        and \eqref{eq:120325a1} and~\eqref{eq:120325a2} in the fifth equality,
        we compute that 
		\begin{align*}
			\rd f (\s (\widehat{Y}_t)) &= f'_- (\s (\widehat{Y}_t)) \vd \s (\widehat{Y}_t) + \frac{1}{2} \int_\bR \rd L^x_t (\s (\widehat{Y})) \vd f'_+ (x)
			\\&= f'_- (\s (\widehat{Y}_t)) \vd \s (\widehat{Y}_t) + \int_\bR \vd L^{\s (x)}_t (\s (\widehat{Y})) g (\s (x))\, \m (\rd x) 
			\\&= f'_- (\s (\widehat{Y}_t)) \vd \s (\widehat{Y}_t) + \int_\bR  \s_+' (x) \vd L^x_t (\widehat{Y}) g (\s (x))\, \m (\rd x) 	
			\\&= f'_- (\s (\widehat{Y}_t)) \vd \s (\widehat{Y}_t) + \frac{\s_+' (\widehat{Y}_t) g (\s (\widehat{Y}_t))}{v_\kappa (\widehat{Y}_t)} \vd \langle \widehat{Y}, \widehat{Y}\rangle_t + c\, \s_+' (\xi) g(\s (\xi)) \vd L^\xi_t (\widehat{Y}) 
			\\&= f'_- (\s (\widehat{Y}_t)) \vd \s (\widehat{Y}_t) + \frac{\s_+' (\widehat{Y}_t) g (\s (\widehat{Y}_t))}{v_\kappa (\widehat{Y}_t)} \1_{\{\widehat{Y}_t \not = \xi\}} \, \rd t+ \frac{c\, \s_+' (\xi) \,g(\s (\xi))}{c\, (1 - \kappa)} \1_{\{\widehat{Y}_t = \xi\}} \, \rd t
			\\&=  f'_- (\s (\widehat{Y}_t)) \vd \s (\widehat{Y}_t) + g (\s (\widehat{Y}_t)) \1_{\{\widehat{Y}_t \not = \xi\}} \, \vd t+ g(\s (\xi)) \1_{\{\widehat{Y}_t = \xi\}} \vd t\phantom \int
			\\&= f'_- (\s (\widehat{Y}_t)) \vd \s (\widehat{Y}_t) + g (\s (\widehat{Y}_t)) \, \rd t . \phantom \int
		\end{align*}
        Therefore, the process
        $$
        f(\s(\widehat Y))-f(\s(x_0))-\int_0^\cdot
        g(\s(\widehat Y_u))\vd u
        $$
        is a local martingale.
	Now it follows from \cite[Theorem~B.3]{CU22} and \cite[Theorem~75, p. 131]{freedman} (or see \cite[Lemma~B.4]{CU22} for a restatement) that \(\s (\widehat{Y})\) has the speed measure \(\widehat\m\). The proof is complete.
	\end{proof}

		\bibliography{scigenbibfile}
		\bibliographystyle{abbrv}

\end{document}